\documentclass[12pt]{article}
\usepackage{amsmath, amssymb, amsfonts, amsthm}
\usepackage{a4wide}
\newcommand{\R}{\mathbb R}

\newcommand{\cQ}{\mathcal Q}

\newcommand{\cE}{\mathcal E}
\newcommand{\tr}{{\rm Tr\;}}

\newcommand{\be}{\begin{equation}}
\newcommand{\ee}{\end{equation}}

\newcommand{\al}{\alpha}

\newcommand{\wt}{\widetilde}

\newtheorem{theorem}{Theorem}[section]

\newtheorem{lemma}[theorem]{Lemma}

\numberwithin{equation}{section}

\DeclareMathOperator{\supp}{supp}

\DeclareMathOperator{\Spec}{Spec}

\title{Relativistic Scott correction in self-generated magnetic fields}

\author{L\'aszl\'o Erd\H os
 \thanks{Partially supported by SFB-TR12 of
the German Science Foundation. {\text lerdos@math.lmu.de} }
\\Institute of Mathematics, University of Munich \\
Theresienstr. 39, D-80333 Munich, Germany \\
S\o ren Fournais \thanks{Work partially supported by the Lundbeck
  Foundation, the Danish Natural Science Research Council and the European 
Research Council under the
 European Community's Seventh Framework Program (FP7/2007--2013)/ERC grant
 agreement  202859.
{\text fournais@imf.au.dk}} \\ Department of Mathematical Sciences, Aarhus University\\
 Ny Munkegade 118, DK-8000 Aarhus, Denmark
\\ and \\
Jan Philip Solovej \thanks{Work partially supported
   by the Danish Natural Science Research Council and by a Mercator
   Guest Professorship from the German Science Foundation. {\text
solovej@math.ku.dk}}
\\ Department of Mathematics, University of Copenhagen\\
Universitetsparken 5, DK-2100 Copenhagen,
Denmark}

\begin{document}

\date{January 31, 2012}
\maketitle

\centerline{\it Dedicated to the 80-th birthday of Elliott H. Lieb}

\begin{abstract}
We consider a large neutral molecule with total nuclear charge $Z$ in a model
 with self-generated classical magnetic field and where the 
kinetic energy of the electrons is treated relativistically.
To ensure stability, we assume that $Z \alpha < 2/\pi$, where 
$\alpha$ denotes the fine structure constant. We are interested in
the ground state energy in the simultaneous
 limit $Z \rightarrow \infty$, $\alpha \rightarrow 0$ such that $\kappa=Z \alpha$ is fixed.
The leading term in the energy asymptotics is independent of $\kappa$, it is
given by the Thomas-Fermi energy
of order $Z^{7/3}$ and it is unchanged by including 
the self-generated magnetic field.
We prove 
the first correction term  to this energy, the so-called Scott correction
of the form $S(\alpha Z) Z^2$.
The current paper extends the result of \cite{SSS} on
the Scott correction for relativistic molecules to include a self-generated
magnetic field. Furthermore, we show that the corresponding Scott correction
function $S$, first identified in \cite{SSS}, is unchanged
by including a magnetic field. We also prove new Lieb-Thirring inequalities
for the relativistic kinetic energy with magnetic fields.
\end{abstract}



\medskip\noindent
{\it Running title:} Relativistic Scott correction


\section{Introduction and results}

We consider a relativistic model of a molecule
in three dimensions, 
where the kinetic energy of the electrons is modelled by the square root of the Pauli operator. 
The nuclei are fixed at positions ${\bf R}=(R_1,\ldots, R_M)$ and 
have charges ${\bf Z} = (Z_1,\ldots, Z_M)$, $Z_k>0$. Let $Z = \sum_{j=k}^M Z_k$
be the total nuclear charge. For simplicity we consider a neutral molecule, i.e.
the number of electrons $N$ is set to be equal to the total nuclear charge, $N=Z$.
The particles are subject to Coulomb interaction and the electrons are dynamical.
The kinetic energy operator of a single electron is
\begin{align}
{\mathcal T}^{(\al)}(A):=  \sqrt{ \alpha^{-2} T(A)+ \alpha^{-4}} - \alpha^{-2},
\end{align}
where $\al>0$ is a parameter (fine structure constant)
and $T(A)$ is the  non-relativistic kinetic energy operator
given by 
\begin{align}\label{Tdef}
T(A): =\begin{cases}
 [ \sigma \cdot (-i\nabla+A)]^2 & \text{ (Pauli) }\\
 (-i \nabla+A)^2 & \text{ (Schr\"{o}dinger).}
\end{cases}
\end{align}
Here $A$ is the magnetic vector potential
generating the magnetic field $B=\nabla\times A$ and $\sigma$ is the vector of 
the three Pauli matrices.
Note that in the  $\alpha\to 0$ limit ${\mathcal T}^{(\al)}(A)$
is replaced with $\frac{1}{2} T(A)$, its non-relativistic counterpart.
We will treat the Pauli case (with spin-$\frac{1}{2}$) and the spinless
Schr\"odinger case in parallel. For simplicity, we write the proofs
for the more difficult Pauli case; the necessary modifications
for the Schr\"odinger case are straight-forward and left to the reader.
 The operator $T(A)$ acts on $L^2({\mathbb R}^3, {\mathbb C}^2)$; in the Schr\"odinger case $T(A)$ is diagonal in the spin variables.

The Hamiltonian of the molecule is 
\begin{align}
\label{eq:H}
H({\bf Z},{\bf R}, \alpha, A) &:= \sum_{j=1}^Z \Big({\mathcal T}_j^{(\al)}(A) -
 \sum_{k=1}^M \frac{Z_k}{|x_j-R_k|}\Big)+ \sum_{j<k} \frac{1}{|x_j-x_k|} ,
\end{align}
where ${\mathcal T}_j^{(\al)}(A)$ acts in the Hilbert space  of the
$j$-th electron. The Hilbert space for the whole system is
$$
{\mathcal H} = \bigwedge_{j=1}^Z L^2({\mathbb R}^3, {\mathbb C}^2).
$$

Our units are $\hbar^2(me^2)^{-1}$ for the length, $me^4\hbar^{-2}$
for the energy and $mec\hbar^{-1}$ for the magnetic vector
potential, where $m$ is the electron mass, $e$ is the electron charge and
$\hbar$ is the Planck constant.
 In these units, the only physical parameter
that  appears in the total  Hamiltonian \eqref{eq:H}
 is the dimensionless fine structure constant $\al=e^2(\hbar c)^{-1} \sim \frac{1}{137}$.
It is known that $\max_k Z_k\alpha \le 2/\pi$ is necessary for   the stability of the system, even
without magnetic field ($A=0$). In this paper
we will assume that $\max_k Z_k\alpha < 2/\pi$ and we will investigate the simultaneous limit $Z\to \infty$, $\alpha\to 0$.

For a given vector potential $A$, the ground state energy of the electrons is given by
\begin{align}
E_0({\bf Z},{\bf R}, \alpha, A) &:=  \inf \Spec H({\bf Z},{\bf R},\alpha, A).
\end{align}
The total energy with a self-generated magnetic field is obtained by adding the field energy 
$(8\pi)^{-1} \alpha^{-2}\int|\nabla \times A|^2$ and minimizing over all vector potentials,
\begin{align}\label{infA}
E_0({\bf Z},{\bf R},\alpha) &:= \inf_{A} \Big\{E_0({\bf Z},{\bf R},\alpha, A) +
 \frac{1}{8\pi \alpha^2} \int_{ {\mathbb R}^3} |\nabla \times A|^2 \Big\}.
\end{align}
Since the magnetic energy  will always be finite, we can also  assume 
that $A\in L^6({\mathbb R}^3)$ (see Appendix of \cite{FLL} for the existence
of such a gauge), and we thus have 
\be 
   \nabla\cdot A =0,  \qquad C^{-1} \Big(\int_{{\mathbb R}^3}  A^6\Big)^{1/3} 
\le  \int_{{\mathbb R}^3} |\nabla \otimes A|^2=
 \int_{{\mathbb R}^3} |\nabla\times A|^2
\label{gauge}
\ee
by the Sobolev inequality, 
where $|\nabla \otimes A|^2 = \sum_{i,j=1}^3 |\partial_i A_j|^2$. 
We will call a vector potential $A$ {\it admissible} if
$A\in L^6({\mathbb R}^3)$, $\nabla \otimes A \in L^2({\mathbb R}^3),$ and $\nabla\cdot A=0$. 
Thus \eqref{infA} can be reformulated as
\begin{align}\label{infA1}
E_0({\bf Z},{\bf R},\alpha) &= \inf_{A} \Big\{E_0({\bf Z},{\bf R},\alpha, A) +
 \frac{1}{8\pi \alpha^2} \int_{ {\mathbb R}^3} |\nabla \otimes A|^2 \Big\},
\end{align}
where the minimization is taken over  all admissible $A$.

\medskip

The main question is the ground state energy in the large $Z$ limit.
The answer depends on whether relativistic or non-relativistic
models are considered and whether magnetic fields are included or not.

In the {\it non-relativistic case without magnetic field} $(A=0)$
the ground state energy to leading term 
 is of order $Z^{7/3}$ and it is given by the Thomas-Fermi theory \cite{LS}.
The next order term, known as the Scott correction, 
is of order $Z^2$ and  it is explicitly given by
\be
     2\cdot \frac{1}{4}\sum_{k=1}^M Z_k^2  
\label{Z2}
\ee
(the additional factor $2$ is due to the spin degeneracy)
and it was rigorously proved for atoms in \cite{H, SW1}
and for molecules in \cite{IS}, see also \cite{SS}.

The ground state energy of the {\it relativistic molecule  without magnetic field} 
up to subleading order (Scott correction) has 
been studied in \cite{SSS} (an alternative proof  for
the special case of atoms, $M=1$, was given in \cite{FSW1}):

\begin{theorem}[Non-magnetic relativistic Scott correction \cite{SSS}] \label{theoremSSS} ~\\
  Let ${\bf z}=(z_1,\ldots,z_M)$ with $z_1,\ldots,z_M>0$,
  $\sum_{k=1}^M z_k=1$, and ${\bf r}=(r_1,\ldots,r_M)\in\R^{3M}$ with
  $\min_{k\ne\ell}|r_k-r_\ell|>r_0$ for some $r_0>0$ be given. Define
  ${\bf Z}=(Z_1,\ldots,Z_M)=Z{\bf z}$ and ${\bf R}=Z^{-1/3}{\bf
    r}$. Then there exist a constant $E^{\rm TF}({\bf z},{\bf r})$ and
  a universal (independent of ${\bf z}$, ${\bf r}$ and $M$)
  continuous, non-increasing function $S_2:[0,2/\pi]\to\R$ with
  $S_2(0)=1/4$ such that as $Z=\sum_{k=1}^M Z_k\to\infty$ and
  $\alpha\to0$ with $\max_k\{Z_k\alpha\}\leq 2/\pi$ we have
\begin{equation} \label{main expansion}
   E_0({\bf Z},{\bf R};\alpha, A=0) = Z^{7/3}E^{\rm TF} ({\bf z},{\bf r}) +
   2\sum_{k=1}^{M}Z_k^2 S_2(Z_k\alpha)
  + {\mathcal O}(Z^{2-1/30})\,.
\end{equation}
The implicit constant in the error term depends only on $M$ and $r_0$.
\end{theorem}

In the recent paper \cite{EFS3}  (see also \cite{EFS1} and \cite{EFS2})
the Scott correction for a {\it non-relativistic} molecule
 in the presence of a self-generated magnetic field was proved and shown to be of the
 form $S_1(\alpha^2 Z) Z^2$, i.e to depend on $\alpha$ through the combination $\alpha^2 Z$
 (see also \cite{Iv, Iv1, Iv2} for an alternative derivation).
 This parameter, $\alpha^2 Z$, is also the parameter which in non-relativistic molecules
 with self-generated magnetic field has to be small to ensure stability.
Note that the physical units chosen in \cite{EFS3} differ by a factor 2
from the choice we made in this paper,  in particular the non-relativistic
$\al\to0$
limit of $T(A)$ is $-\frac{1}{2}\Delta$ in the current units,
so the Thomas-Fermi energy is modified compared with  \cite{EFS3}.
Moreover, the notation for
the Scott function $S_1$ incorporates the $8\pi$ factor explicitly appeared
in (1.8) of \cite{EFS3}. The notations in the current paper follow
the conventions of \cite{SSS}.

In the light of these previous results, it is natural to ask the following questions 
for relativistic molecules with self-generated field.
\begin{enumerate}
\item (Existence of Scott term)\\ Is it true that there exists a function $S_3$ such that
\begin{align}
E_0({\bf Z},{\bf R},\alpha)  = Z^{7/3}E^{\rm TF} ({\bf z},{\bf r}) +
   2\sum_{k=1}^{M}Z_k^2 S_3(Z_k\alpha) + o(Z^2),  
\end{align}
in the simultaneous limit $Z \rightarrow \infty$, $\alpha\to0$
with $\kappa = Z \alpha$ fixed, 
for any $\kappa$ small?
\item (The Scott term is non-magnetic)\\
Is it true that $S_2 = S_3$, i.e. the Scott term with self-generated 
magnetic field is the same 
as for the non-magnetic operator ($A=0$)?
\end{enumerate}

The following main theorem of this paper gives an affirmative answer to these questions:

\begin{theorem}[Relativistic Scott correction with self-generated field]
 \label{maintheorem} ~\\
Let the assumptions and notations be as in Theorem~\ref{theoremSSS},
in particular we  fix $M$, ${\mathbf z}$ and  ${\mathbf r}$.
Assume furthermore that there exists $\kappa_0< 2/\pi$ such that
\begin{align}
\max_k\{ Z_k \alpha \} \le \kappa_0.
\end{align}
Then the ground state energy with self-generated magnetic field is given by
\begin{equation} \label{mainselfgen}
   E_0({\bf Z},{\bf R};\alpha) = Z^{7/3}E^{\rm TF} ({\bf z},{\bf r}) +
   2\sum_{k=1}^{M}Z_k^2 S_2(Z_k\alpha)
  + o(Z^{2})\,
\end{equation}
in the limit as $Z\to\infty$ and $\alpha\to 0$. 
\end{theorem}

In contrast to the non-relativistic case \cite{EFS3}, the Scott correction is
non-magnetic in the relativistic case. The reason is that
the prefactor $(8\pi \alpha^2)^{-1}$ in front of the magnetic energy
is of order $Z^2$ in the relativistic case (since $Z\alpha$ is bounded),
 i.e. it is much larger than in the non-relativistic case (when  $Z\alpha^2$ was bounded).
Therefore the self-generated magnetic field is much smaller in the relativistic
case and it eventually does not influence the kinetic energy up to the
order of the Scott term. 
In fact,  our proof gives a somewhat stronger
result; it proves that Theorem~\ref{maintheorem} also holds 
if the constant $8\pi$ in \eqref{infA1} is replaced with any
fixed positive finite number.

\medskip

We now comment on the new ingredients of the proof.
Since the magnetic field is not expected to influence the final
result, we can treat it perturbatively. To control this perturbation, our main  tools
are:  i) a new magnetic Lieb-Thirring type inequality for the relativistic 
case; and ii) a new localization scheme for the kinetic energy operator
${\mathcal T}^{(\alpha)}(A)$.

The magnetic Lieb-Thirring inequality for the non-relativistic Pauli
operator has been proven in \cite{LLS}, while the Daubechies inequality
handles the relativistic case without magnetic field \cite{Dau}.
Our Theorem~\ref{thm:LT} combines and generalizes these two classical inequalities.
We also need a modified version of this result that allows us to include
Coulomb singularities with subcritical coupling constants (Theorem~\ref{thm:CritStability}).
We remark  that magnetic fields have been incorporated
into the Daubechies inequality even with the critical Coulomb singularity \cite{FLS},
but this  result concerns only  the
  Schr\"odinger case \cite{FLS} where diamagnetic techniques are available.

The main localization formula used in \cite{SSS}
 (Theorem 2.5) is not applicable with a magnetic field since 
it relies on the explicit formula for the relativistic heat kernel.
Instead, we use the usual IMS formula under the square root,
then apply the operator-monotonicity of the square root function
and a  useful ``Pull-out'' inequality (Lemma~\ref{lem:Pull-out}).
Constantly adjusting the parameter $\alpha$ in ${\mathcal T}^{(\alpha)}(A)$,
we can show that the localization errors can be controlled
essentially as effectively as in \cite{SSS} despite the lack
of any explicit formula.

\section{Structure of the proof}

The main steps of the proof of Theorem~\ref{maintheorem} follow 
the proof  of Theorem~\ref{theoremSSS} given in \cite{SSS}. To avoid
unnecessary repetitions, we will sometimes explicitly refer to certain
lemmas from  \cite{SSS}, but otherwise we keep the current
paper self-contained.  We will focus 
on the modifications needed in order to accomodate the self-generated magnetic field.

\medskip

We consider the number of nuclei $M$ and the minimal distance $r_0$ among
the rescaled nuclear centers to be fixed throughout the proof and
every generic constant denoted by $C$ in the sequel may depend on them.
The notations 
$$
  [a]_+:=\max\{0, a\}\ge 0, \qquad  [a]_-:=\min \{0,a\}\le 0
$$
 stand
for the positive and  negative parts of a real number or a self-adjoint operator $a$.
Integrals with unspecified integration domain are always considered on ${\mathbb R}^3$.

The upper bound in \eqref{mainselfgen} follows from \eqref{main expansion} by choosing $A=0$
in \eqref{infA}.
So we only need to consider the lower bound.

\subsection{Passage to the mean field Thomas-Fermi theory}

We will use the Thomas-Fermi theory for non-relativistic molecules
 without magnetic field \cite{LS}.  We will not introduce this theory here
in details, we refer the reader
to Section 2.7 of \cite{SSS} whose notation we follow. In particular,
let $V^{TF}_{{\mathbf Z}, {\mathbf R}}(x)=V^{TF}({\mathbf Z}, {\mathbf R}, x)$ 
be the  Thomas-Fermi potential
and $\rho^{TF}_{{\mathbf Z}, {\mathbf R}}(x)=\rho^{TF}({\mathbf Z}, {\mathbf R}, x) $
 the corresponding Thomas-Fermi density and let
$$
   D(f,g) = \frac{1}{2}\int\!\!\int \frac{\overline{f(x)} g(y)}{|x-y|} dx dy, \qquad D(f):= D(f,f).
$$

Define the functions
\begin{align}\label{ddef}
d_{{\bf r}}(x) &= \min_{k=1,\ldots,M} \{|x-r_k|\}, \\
 d_{{\bf R}}(x) &= \min_{k=1,\ldots,M} \{|x-R_k|\} = Z^{-1/3} d_{{\bf r}}(Z^{1/3}x).
\end{align}
The Thomas-Fermi potential $V^{TF}_{{\mathbf z}, {\mathbf r}}(x)$ satisfies the
following bounds for all multi-indices $n\in {\mathbb N}^3$ and all $x$
with $d_{{\bf r}}(x)\ne 0$:
\begin{align}\label{VTFbound}
  \big|\partial_x^n V^{TF}_{{\mathbf z}, {\mathbf r}}(x)\big| \le C^*_n \min\{ d_{{\bf r}}(x)^{-1}, d_{{\bf r}}(x)^{-4}\}
 d_{{\bf r}}(x)^{-|n|},
\end{align}
and we also have
\begin{align}\label{VTFbound1}
  \Big| V^{TF}_{{\mathbf z}, {\mathbf r}}(x) -\frac{z_k}{|x-r_k|}\Big| \le C^*, \quad \mbox{for}\quad 
|x-r_k|\le r_0/2, \quad k=1,2,\ldots M,
\end{align}
where the constants $C^*_n$ and $C^*$ depend only on  $r_0$, $M$ and $\max \{ Z_1, Z_2, \ldots,Z_M\}$
(see Theorem 2.12 and Remark 2.14 in \cite{SSS}).

We get from the correlation estimate \cite[Theorem~2.9 (see also calculation on p. 55)]{SSS}
 that if $\psi \in {\mathcal H}$ is normalized, then
\begin{align}\label{eq:corr}
\langle \psi, H({\mathbf Z}, {\mathbf R}, \alpha, A) \psi \rangle 
&\geq
\tr\big[ {\mathcal T}^{(\al)}(A) - V^{TF}_{{\mathbf Z}, {\mathbf R}}(x)
 - C Z^{3/2} s g(x) \big]_{-}\nonumber \\
&\quad- D(\rho^{TF}_{ {\mathbf Z}, {\mathbf R}}) - Cs Z^{8/3} - C s^{-1} Z,
\end{align}
where the parameter $s$ is chosen as 
$$
 s=Z^{-5/6}
$$
 and 
\begin{align}
g(x) = \begin{cases}
(2s)^{-1/2}, & d_{{\mathbf R}}(x) < 2s,\\
 d_{{\mathbf R}}(x)^{-1/2}, & 2s \leq  d_{{\mathbf R}}(x) \leq Z^{-1/3},\\
 0,  &Z^{-1/3}< d_{{\mathbf R}}(x). 
\end{cases}
\end{align}
With the above choice of $s$, we have $s Z^{8/3}=  s^{-1} Z = Z^{11/6}$,
so the last two terms in \eqref{eq:corr} are negligible to the order $o(Z^2)$
we are interested.

We will estimate the error term $C Z^{3/2} s g(x) = CZ^{2/3}g(x)$ by borrowing 
a small $\delta$-part 
 of the kinetic energy and using the Lieb-Thirring inequality, Theorem~\ref{thm:LT} below
(with $h=1$ and $\beta=\alpha$). 
With the choice of  $\delta = Z^{-1/2}$, using  $\alpha \leq C Z^{-1}$
and computing $\int g^{5/2}\le CZ^{-7/12}$ and $\int g^4 \le CZ^{-1/3}$, we get
\begin{align}
\tr\Big[ \delta {\mathcal T}^{(\al)}(A) - C Z^{3/2} s g(x) 
\Big]_{-} 
\geq
- C Z^{11/6} - C Z^{7/12} \Big( \int |\nabla \times A|^2 \Big)^{3/4}.
\end{align}
Inserting these estimates in \eqref{eq:corr}
and using that $\tr[X+Y]_-\ge \tr [X]_-+ \tr[Y]_-$ for self-adjoint operators $X$, $Y$,
 we get for any $A$ and any normalized $\psi \in {\mathcal H}$,
\begin{align}\label{eq:corr2}
\langle \psi, H({\mathbf Z}, {\mathbf R}, \alpha, A) \psi \rangle
&\geq
\tr\big[ (1- Z^{-1/2}) {\mathcal T}^{(\al)}(A)  - V^{TF}_{{\mathbf Z}, {\mathbf R}}  \big]_{-}\nonumber \\
&\quad- D(\rho^{TF}_{{\mathbf Z}, {\mathbf R}}) - C Z^{11/6} - C Z^{7/12} \Big( \int |\nabla \times A|^2
 \Big)^{3/4} \nonumber \\
&\geq \tr\big[ (1- Z^{-1/2}) {\mathcal T}^{(\al)}(A) - 
V^{TF}_{{\mathbf Z}, {\mathbf R}}  \big]_{-}\nonumber \\
&\quad- D(\rho^{TF}_{{\mathbf Z}, {\mathbf R}}) - C (Z^{11/6} + Z^{7/3} \alpha^6)
 - \frac{1}{16 \pi \alpha^2} \int |\nabla \times A|^2.
\end{align}
So, using $Z \alpha \le C$,
\begin{align}\label{eq:corr3}
\langle \psi, H({\mathbf Z}, {\mathbf R}, \alpha, A) \psi \rangle+
\frac{1}{8 \pi \alpha^2} \int |\nabla \times A|^2\geq &\;\tr\big[ (1- Z^{-1/2}){\mathcal T}^{(\al)}(A)
  - V^{TF}_{{\mathbf Z}, {\mathbf R}}  \big]_{-}
\\
&- D(\rho^{TF}_{{\mathbf Z}, {\mathbf R}} ) - C Z^{11/6} + 
\frac{1}{16 \pi \alpha^2} \int |\nabla \times A|^2.
\nonumber
\end{align}
Clearly, the constant 16 can be replaced with any finite constant larger than $8$
at the expense of changing $C$ in the last line.

\subsection{Scaling}

We now introduce the usual semiclassical scaling of the
Thomas-Fermi theory. 
The kinetic energy operator with the semiclassical parameter $h$ is defined by
\begin{align}
T_h(A) =\begin{cases}
 [ \sigma \cdot (-ih \nabla+A)]^2 & \text{ (Pauli) }\\
 (-ih \nabla+A)^2 & \text{ (Schr\"{o}dinger)}
\end{cases}
\end{align}
and clearly $T(A)$ from \eqref{Tdef} equals to
$T_{h=1}(A)$.

Define
\begin{align}
\kappa = \min_k \frac{2}{\pi z_k},\qquad
h= \kappa^{1/2} Z^{-1/3},\qquad
\beta = Z^{2/3} \alpha  \kappa^{-1/2} = \frac{Z\alpha}{\kappa} h.
\label{param}
\end{align}
In particular, since $Z_k \alpha \leq 2/\pi$, we have $\beta \leq h$.
Note that the  notation generally follows \cite{SSS}, but our
 definition  of $\beta$ differs from \cite{SSS} by a square root.

The Thomas-Fermi potential and density satisfy the scaling relation
\begin{align}\nonumber
V^{TF}_{{\mathbf Z}, {\mathbf R}}(x) = a^{4} V^{TF}_{a^{-3}{\mathbf Z}, a{\mathbf R}}(ax),
 \qquad \rho^{TF}_{{\mathbf Z}, {\mathbf R}}(x) = a^{6} \rho^{TF}_{a^{-3}{\mathbf Z}, a{\mathbf R}}(ax)
\end{align}
for any $a>0$.
In particular, 
\begin{align}\label{scaling}
V^{TF}_{{\mathbf Z}, {\mathbf R}}(x) = Z^{4/3} V^{TF}_{{\mathbf z}, {\mathbf r}}(Z^{1/3} x), \qquad
  D( \rho^{TF}_{{\mathbf Z}, {\mathbf R}}) = Z^{7/3} D( \rho^{TF}_{{\mathbf z}, {\mathbf r}}).
\end{align}
We will perform the scaling $x \mapsto Z^{-1/3} x$. During the scaling we 
replace the vector potential $A$ by 
$$
\widetilde A(x): = Z^{-2/3} \kappa^{1/2}A(Z^{-1/3} x),
$$
 so we get for
 the magnetic energy in \eqref{eq:corr3} 
$$
\frac{1}{16 \pi \alpha^2} \int |\nabla \otimes A|^2 
= Z^{4/3} \Big\{ \frac{1}{16 \pi \beta^2 h^3\kappa^{1/2}} \int |\nabla \otimes \widetilde A|^2 \Big\}
$$ 
and $T(A)$ is replaced with 
$T_h(\wt A)$.
Using \eqref{eq:corr3} and \eqref{scaling} we therefore get
\begin{align}\label{eq:BeforeSemiclass}
\langle \psi, &H({\mathbf Z}, {\mathbf R}, \alpha, A) \psi \rangle+\frac{1}{8 \pi \alpha^2} \int |\nabla \times A|^2\nonumber \\
&\geq
Z^{4/3} \kappa^{-1} (1-Z^{-1/2})
\Big\{
\tr\big( \sqrt{\beta^{-2} T_{h}(\wt A) + \beta^{-4}} - \beta^{-2} - \frac{\kappa}{1-Z^{-1/2}} V^{TF}_{{\mathbf z}, {\mathbf r}} \big)_{-}  \nonumber\\
&\qquad\qquad\qquad\qquad\qquad\qquad+ \frac{\kappa^{1/2}}{16 \pi \beta^2 h^3} \int |\nabla \otimes \widetilde A|^2
\Big\} \nonumber \\
&\quad- Z^{7/3}D( \rho^{TF}_{{\mathbf z}, {\mathbf r}}) - C Z^{11/6}.
\end{align}
The proof of the Scott correction is now reduced to the proof of the following semiclassical theorem:

\begin{theorem}[Scott corrected semiclassics with self-generated field]\label{thm:Thm1.4bis}~\\
Suppose that $\lambda>0$.
There exists a function $\xi: {\mathbb R}_+ \rightarrow  {\mathbb R}_+$ with 
$\xi(h) \rightarrow 0$ as $h \rightarrow 0$, such that if
$0\leq \beta \leq h$, and $\wt \kappa \max\{ z_1,\ldots,z_M\} < 2/\pi$,
then
\begin{align}
&\Big| \inf_{\wt A} \Big\{ \tr\Big[\sqrt{\beta^{-2} T_{h}(\wt A) + \beta^{-4}} - \beta^{-2}
 - \wt \kappa V^{TF}_{{\mathbf z}, {\mathbf r}} \Big]_{-} +
 \frac{\lambda}{\beta^2 h^3} \int |\nabla \otimes \widetilde A|^2
\Big\} \nonumber \\
&\qquad\qquad - \frac{2}{(2\pi h)^3} \iint \Big[ \frac{1}{2}p^2 -  \wt \kappa V^{TF}_{{\mathbf z}, {\mathbf r}}(x)
 \Big]_{-} \,dx dp
- 2 h^{-2} \sum_{k=1}^M (z_k \wt \kappa)^2 S_2(\beta h^{-1} \wt\kappa z_k) \Big| \nonumber \\
&\quad\leq h^{-2} \xi(h).
\label{eq:scottcorr}
\end{align}
\end{theorem}
Notice that the semiclassical asymptotics up to the subleading Scott term is
independent of the parameter $\lambda$ in front of the magnetic field energy.
This justifies the remark that the specific constant $8\pi$ in \eqref{infA1} is
irrelevant and can be replaced with any positive constant.

\medskip

We will first finish the proof of Theorem~\ref{maintheorem}
using Theorem~\ref{thm:Thm1.4bis} before giving the proof of the semiclassical result.
Using Theorem~\ref{thm:Thm1.4bis} in \eqref{eq:BeforeSemiclass} and the
choice of the parameters \eqref{param}, we get
\begin{align}
\langle \psi, &H({\mathbf Z}, {\mathbf R}, \alpha, A) \psi \rangle+\frac{1}{8 \pi \alpha^2} \int |\nabla \times A|^2\nonumber \\
&\geq
Z^{4/3} \kappa^{-1} (1-Z^{-1/2})
\Big\{\frac{2}{(2\pi h)^3} \iint \Big[ \frac{1}{2}p^2 -  
\frac{\kappa}{1-Z^{-1/2}}  V^{TF}_{{\mathbf z}, {\mathbf r}}(x) \Big]_{-} \,dx dp\nonumber \\
&\qquad\qquad+2 h^{-2} \sum_{k=1}^M (\frac{z_k\kappa}{1-Z^{-1/2}} )^2 S_2( \frac{\beta h^{-1}z_k \kappa}{1-Z^{-1/2}} )
- h^{-2} g(h)
\Big\} - D(\rho^{TF}) - C Z^{11/6} \nonumber \\
&=
Z^{7/3} E^{TF}({\bf z},{\bf r}) + 2 \sum_{k=1}^M Z_k^2S_2(Z_k \alpha) + o(Z^2).
\end{align}
Here we used the continuity of $S_2$ and the facts that 
$$\iint \Big[ \frac{1}{2}p^2 - V^{TF}_{{\mathbf z}, {\mathbf r}}(x) \Big]_{-} \,dx dp = 
C \int [V^{TF}_{{\mathbf z}, {\mathbf r}}(x)]^{5/2} dx < \infty
$$
 and
$$
Z^{7/3} \frac{2}{(2\pi)^3} \iint \Big[ \frac{1}{2}p^2 - V^{TF}_{{\mathbf z}, {\mathbf r}}(x)
 \Big]_{-} \,dx dp - Z^{7/3} D(\rho^{TF}_{{\mathbf z}, {\mathbf r}}) =
Z^{7/3} E^{TF}({\bf z},{\bf r}),
$$
by standard results from Thomas-Fermi theory.
This finishes the proof of Theorem~\ref{maintheorem}. $\;\;\Box$

\subsection{Relativistic Lieb-Thirring inequalities with magnetic fields}

In this section we present two new Lieb-Thirring type inequalities
for the relativistic kinetic energy with a magnetic field.
The proofs are given in Section~\ref{sec:proof}.
\begin{theorem}[Lieb-Thirring inequality for ${\mathcal T}^{(\beta)}(A)$]\label{thm:LT}
There exists a universal constant $C>0$ such that
for any positive number $\beta>0$,
for any potential $V$ with $[V]_{+} \in L^{5/2}\cap L^{4}({\mathbb R}^3)$, 
and magnetic field $B = \nabla \times A \in L^2({\mathbb R}^3)$, we have
\begin{align}
&\tr \big[ \sqrt{\beta^{-2} T(A) + \beta^{-4}} - \beta^{-2} - V(x) \big]_{-} \nonumber \\
&\qquad\geq-
C\Bigg\{ \int [V]_{+}^{5/2} + \beta^3  \int [V]_+^4 
+\Big( \int B^2 \Big)^{3/4} \Big(\int [V]_+^4\Big)^{1/4}\Bigg\}.
\label{LT1}
\end{align}
\end{theorem}

Notice that 
Theorem~\ref{thm:LT} reduces to the well-known Daubechies inequality in the case $A=0$ \cite{Dau}.
For the Schr\"odinger case, 
the Daubechies inequality was generalized (and improved to incorporate a critical Coulomb singularity)
 to non-zero $A$   in \cite{FLS} by using diamagnetic techniques. Theorem~\ref{thm:LT}
is the generalization of the Daubechies inequality for the Pauli operator, in which case there is 
no diamagnetic inequality.
Moreover, in the $\beta\to 0$ limit, \eqref{LT1} converges to the magnetic Lieb-Thirring inequality
for the Pauli operator \cite{LLS}
since
$$
    \sqrt{\beta^{-2} T(A) + \beta^{-4}} - \beta^{-2}\to \frac{1}{2} T(A), \qquad \beta\to0.
$$

 Theorem~\ref{thm:LT} does not cover 
the case of a Coulomb singularity.
The next result shows that  for $\beta$ smaller than the
critical value $2/\pi$, the Coulomb singularity can be included.
 The constraint $\beta<2/\pi$ in this theorem is the
main reason why our proof of  Theorem~\ref{maintheorem} 
does not extend to the critical case, $\kappa_0=2/\pi$.

\begin{theorem}[Local Lieb-Thirring inequality with a Coulomb potential]\label{thm:CritStability}
Let $\phi_r$ be a real function satisfying $\supp \phi_r \subset \{|x|\leq r\}$, 
$\| \phi_r \|_{\infty} \leq 1$.
There exists a constant $C>0$ such that if $\beta \in (0,2/\pi)$, then
\begin{align}\label{eq:stability2}
\tr \Big[ \phi_r \big( &\sqrt{\beta^{-2} T(A) + \beta^{-4}} - \beta^{-2} -
 \frac{1}{|x|}- V \big) \phi_r \Big]_{-} \nonumber \\
&\geq
- C \Big\{ \eta^{-3/2} \int |\nabla \times A|^2 + \eta^{-3} r^3 + \eta^{-3/2} \int [V]_{+}^{5/2} + 
\eta^{-3} \beta^3 \int [V]_{+}^4 \nonumber \\
&\qquad\quad+ \Big(\int |\nabla \times A|^2 \Big)^{3/4}  \Big(\int [V]_{+}^4\Big)^{1/4} \Big\},
\end{align}
where $\eta := \frac{1}{10}(1-(\pi\beta/2)^2)$.
\end{theorem}

For simplicity, we stated both theorems for $h=1$, but 
$T(A)=T_{h=1}(A)$ can easily be replaced with $T_h(A)$ and the $h$ scaling on
the right hand sides can be easily computed. In particular, we have from \eqref{LT1} that
\begin{align}
&\tr \big[ \sqrt{\beta^{-2} T_h(A) + \beta^{-4}} - \beta^{-2} - V(x) \big]_{-} \nonumber \\
&\qquad\geq-
C\Bigg\{ h^{-3}\int [V]_{+}^{5/2} + h^{-3}\beta^3  \int [V]_+^4 
+\Big( h^{-2} \int B^2 \Big)^{3/4} \Big(\int [V]_+^4\Big)^{1/4}\Bigg\}.
\label{LT2}
\end{align}

\section{Proof of Theorem~\ref{thm:Thm1.4bis}}
For the upper bound in Theorem~\ref{thm:Thm1.4bis} we can just take $A=0$ and apply \cite[Theorem~1.4]{SSS}.
Notice that this will actually also provide us with a non-magnetic trial state which has the correct energy.

For the lower bound we will follow the proof of the similar result
 in the non-magnetic case \cite[Theorem~1.4]{SSS} (see pages 68--75).
 In particular, we will use the same localizations. 

Consider a smooth partition of unity,
$$
\theta_{-}^2 + \theta_{+}^2 =1,
$$
where $\theta_{-}(t) =1$ if $t<1$, $\theta_{-}(t) =0$ if $t >2$. 
Define
\begin{align}
  \label{eq:4}
\Phi_{\pm}(x) = \theta_{\pm}(d(x)/R), \qquad \phi_{\pm}(x) = \theta_{\pm}(d(x)/r).  
\end{align}
where $d(x):=d_{{\mathbf r}}(x)$ for simplicity.
At the end of the calculation we will see that
the small parameter $r$ can be chosen $r=h^{3/2}$ and the large
parameter $R$ as $R= h^{-1}$. We note that $r$ was chosen differently
 in \cite{SSS}.

For $r>0$ we define
\begin{align}\label{phimin}
\phi_{-}(x) := \sum_{k=1}^M \theta_{r,k}(x), \quad \text{ with } \theta_{r,k}(x) := \theta_{-}(|x-r_k|/r).
\end{align}
and we note that for $r$ sufficiently small and $R$ sufficiently large,
$\theta_{r,k}$ have disjoint supports
and these supports are in the regime where $\Phi_-\equiv 1$.
We thus have 
$$
\sum_k \theta_{r,k}^2+\Phi_-^2\phi_+^2+\Phi_+^2= 1
$$
as a partition of unity.
Defining
$$
   W_{r,R}(x):=  r^{-2} {\bf 1}_{\{ r\le d(x) \leq 2r\}}+R^{-2} {\bf 1}_{\{ R\leq d(x) \leq 2R\}},
$$
the IMS formula allows us to insert these localizations and estimate
\begin{align}
  T_{h}(\wt A)+\beta^{-2}\ge & \sum_{k=1}^M \theta_{r,k}\big(  T_{h}(\wt A) - Ch^2r^{-2}+\beta^{-2}\big) \theta_{r,k}
 \nonumber\\
&  + \Phi_{-} \phi_+ \Big(  T_{h}(\wt A)-
  C h^2 W_{r,R}
 +\beta^{-2}\Big)\Phi_{-} \phi_+
\nonumber\\
& + \Phi_{+}\Big( T_h(\wt A) - C h^2
   R^{-2} {\bf 1}_{\{ d(x) \leq 2R\}}+\beta^{-2}  \Big) \Phi_{+}.
\label{IMS}
\end{align}
Notice that all operators in the brackets in the right hand side are
non-negative, since $h^2R^{-2}\ll h^2r^{-2}=h^{-1}\ll \beta^{-2}$.
After multiplying both sides by $\beta^{-2}$ we can take
the square root of the inequality \eqref{IMS}, using that
the square root  is operator monotone.
In order to  pull out the
localization functions from under the square root, we will need the following general estimate:

\begin{lemma}[Pull-out estimate]\label{lem:Pull-out}
Let $I$ be a countable index set and let $g_i$, $i\in I$, be a family of non-negative 
smooth functions such that $\sum_{i\in I} g_i^2(x) = 1$ for every $x \in {\mathbb R}^3$.
 Let $A_i$, $i \in I$, be a family of positive self-adjoint
 operators on $L^2({\mathbb R}^3, {\mathbb C}^2)$. Then
\begin{align}\label{eq:Pull-out}
\sqrt{ \sum_{i\in I} g_i A_i g_i} \geq \sum_{i \in I} g_i \sqrt{A_i} g_i.
\end{align}
\end{lemma}

\begin{proof}
The proof follows from the integral representation
$$
\sqrt{A} = \mbox{(const.)}\int_0^{\infty} \Big(1 - \frac{t}{A+t}\Big) \frac{dt}{\sqrt{t}},
$$
and from the similar ``pull-up'' formula for the resolvents
which first appeared in \cite{BFFGS}, see also \cite[Proposition~6.1]{ES1}.
\end{proof}

Applying the estimate \eqref{eq:Pull-out}, we get from \eqref{IMS} that
\begin{align}
  \label{eq:5}
  &\sqrt{\beta^{-2} T_{h}(\wt A) + \beta^{-4}} - \beta^{-2} - \wt
  \kappa V^{TF}_{{\mathbf z}, {\mathbf r}}\nonumber \\
&\geq
\sum_{k=1}^M \theta_{r,k} \Big( \sqrt{\beta^{-2} T_{h}(\wt A)-
  C \beta^{-2} h^2 r^{-2} + \beta^{-4}} - \beta^{-2} - \wt
  \kappa V^{TF}_{{\mathbf z}, {\mathbf r}} \Big) \theta_{r,k} \nonumber \\
&\quad +\Phi_{-} \phi_+ \Big( \sqrt{\beta^{-2} T_{h}(\wt A)-
  C \beta^{-2} h^2 W_{r,R}+ \beta^{-4}} - \beta^{-2} - \wt
  \kappa V^{TF}_{{\mathbf z}, {\mathbf r}} \Big)\phi_+\Phi_{-}  \nonumber \\
&\quad+\Phi_{+}\Big\{ \sqrt{\beta^{-2} T_h(\wt A) + \beta^{-4} - C h^2
  \beta^{-2} R^{-2} {\bf 1}_{\{ d(x) \leq 2R\}}} - \beta^{-2} - \wt \kappa V^{TF}_{{\mathbf z}, {\mathbf r}} \Big\} \Phi_{+}
\end{align}
The three terms will be considered independently in the next three subsections.
 The first one contains
the contributions from the nuclei and will give the Scott term. The
second gives the main contribution to the energy and is
semiclassical. Finally the third term is a small error term.
The main technical results, the analysis of the Scott term and the local
semiclassical asymptotics, will be proved separately in Sections~\ref{sec:Scott} and
\ref{sec:locsc}.

\subsection{The region far from the nuclei}
Let us start by considering the outer term in \eqref{eq:5} resulting from the localization procedure, namely
\begin{align}
\tr\Big[\Phi_{+}\Big\{ \sqrt{\beta^{-2} T_h(\wt A) + \beta^{-4} - C h^2
  \beta^{-2} R^{-2} {\bf 1}_{\{ d(x) \leq 2R\}}} - \beta^{-2} - \wt \kappa V^{TF}_{{\mathbf z}, {\mathbf r}} \Big\} \Phi_{+}\Big]_{-}.
\end{align}
We introduce a dyadic partition of unity
\begin{align}
 &\supp \phi_{0,R} \subset B(4R),
\quad \supp \phi_{j,R} \subset B(2^{j+2}R) \setminus B(2^j
R),&\nonumber \\ 
&\sum_{j=0}^{\infty} \phi_{j,R}^2 =1, \quad\qquad \qquad |\nabla \phi_{j,R}| \leq C 2^{-j} R^{-1}.
\end{align}
Then 
\begin{align}
T_h(\wt A) \geq \sum_j \phi_{j,R} \big[ T_h(\wt A) - C h^2 2^{-2j} R^{-2}\big] \phi_{j,R}.
\end{align}
So,
\begin{multline}
\beta^{-2} T_h(\wt A) + \beta^{-4} - C h^2 \beta^{-2} R^{-2} {\bf 1}_{\{ d(x) \leq 2R\}}\\
\geq
\sum_j \phi_{j,R} \big[ \beta^{-2}T_h(\wt A) - C \beta^{-2} h^2 2^{-2j} R^{-2} + \beta^{-4} \big] \phi_{j,R}.
\end{multline}
Therefore, by operator monotonicity of the square root and the pull-out estimate, Lemma~\ref{lem:Pull-out},
\begin{align}
\Phi_{+}&\Big\{ \sqrt{\beta^{-2} T_h(\wt A) + \beta^{-4} -
 C h^2 \beta^{-2} R^{-2} {\bf 1}_{\{ d(x) \leq 2R\}}} - \beta^{-2} - \wt\kappa V^{TF}_{{\mathbf z}, {\mathbf r}} \Big\} \Phi_{+}\nonumber \\
&\geq
\sum_{j=0}^{\infty}
\Phi_{+}\phi_{j,R}\Big\{ \sqrt{\beta^{-2} T_h(\wt A)
 + \beta^{-4} - C h^2 \beta^{-2} 2^{-2j}R^{-2} } - \beta^{-2} - \wt\kappa V^{TF}_{{\mathbf z}, {\mathbf r}} \Big\} \Phi_{+}\phi_{j,R}.
\end{align}
We define 
$$
\gamma_j := \beta (1- C h^2 \beta^{2} 2^{-2j}R^{-2} )^{-1/4},
$$
and note that $\beta\le \gamma_j\le 2\beta\le 2h$ using $\beta\le h$ and $R\ge 1$.
We also have
$0\leq \beta^{-2} - \gamma_j^{-2} \leq C h^2 2^{-2j} R^{-2}$, so we can continue the estimate 
(using the operator monotonicity of the square root) as
\begin{align}
\Phi_{+}&\Big\{ \sqrt{\beta^{-2} T_h(\wt A) + \beta^{-4} - C h^2 \beta^{-2} R^{-2} {\bf 1}_{\{ d(x) \leq 2R\}}}
 - \beta^{-2} - \wt\kappa V^{TF}_{{\mathbf z}, {\mathbf r}} \Big\} \Phi_{+}\nonumber \\
&\geq
\sum_{j=0}^{\infty}
\Phi_{+}\phi_{j,R}\Big\{ \sqrt{\gamma_j^{-2} T_h(\wt A) + \gamma_j^{-4} } - \gamma_j^{-2} - C h^2 2^{-2j} R^{-2}
-  \wt\kappa V^{TF}_{{\mathbf z}, {\mathbf r}} \Big\} \Phi_{+}\phi_{j,R}.
\end{align}
Recall  from \eqref{VTFbound} that $|V^{TF}_{{\mathbf z}, {\mathbf r}}(x)| \leq C d_{{\mathbf r}}(x)^{-4}$
and $\wt\kappa \le 2M/\pi$ since $\wt\kappa \max_k z_k\le 2/\pi$ and $\max_k z_k\ge 1/M$.
Thus $|\wt \kappa V^{TF}_{{\mathbf z}, {\mathbf r}}(x)|\le C|x|^{-4}$ as $|x|\geq C$, where 
the large constant $C$ depends only 
on $M$ and ${\mathbf r}$. Hence $|V^{TF}_{{\mathbf z}, {\mathbf r}}| \leq C
2^{-4j}R^{-4}$ on $\supp \phi_{j,R}$.
With the choice $R=h^{-1}$ we can therefore absorb $V^{TF}_{{\mathbf z}, {\mathbf r}}$ in the
term $C h^2 2^{-2j} R^{-2}$ by changing the constant $C$.

{F}rom the semiclassical form of the Lieb-Thirring inequality \eqref{LT2} we get
\begin{align}
&\tr\Big[ \Phi_{+}\phi_{j,R}\Big\{ \sqrt{\gamma_j^{-2} T_h(\wt A) + \gamma_j^{-4} } - \gamma_j^{-2}
 - C h^2 2^{-2j} R^{-2} \Big\} \Phi_{+}\phi_{j,R}\Big]_{-}\nonumber \\
&\geq -C \Big\{ h^{-3} 2^{3j} R^3 \big( h^5 2^{-5j} R^{-5} + \gamma_j^3 h^8 2^{-8j} R^{-8}
  \big) + \Big(h^{-2} \int |\nabla \otimes \wt A|^2\Big)^{3/4} ( 2^{-5j} R^{-5} h^8 )^{1/4}\Big\}
 \nonumber \\
&\geq
-C \Big\{ h^2 R^{-2} 2^{-2j} +  \gamma_j^3 h^5 2^{-5j} R^{-5}+2^{-2j} h^5 R^{-5} +
2^{-j} h^{-1} \int |\nabla \otimes \wt A|^2\Big\}.
\end{align}
Using the trivial estimate $\gamma_j  \le 2h\le 1$ and summing up, we therefore find
\begin{align}
&\sum_j 
\tr\Big[\Phi_{+}\phi_{j,R}\Big\{ \sqrt{\gamma_j^{-2} T_h(\wt A) + \gamma_j^{-4} }
 - \gamma_j^{-2} - C h^2 2^{-2j} R^{-2} \Big\} \Phi_{+}\phi_{j,R}\Big]_{-}\nonumber \\
&\geq
- C \Big\{ h^2 R^{-2} +  h^5  R^{-5}+ h^{-1} \int |\nabla \otimes \wt A|^2\Big\}
\end{align}
With the choice $R= h^{-1}$ we  get
\begin{align}
&\sum_j 
\tr\Big[ \Phi_{+}\phi_{j,R}\Big\{ \sqrt{\gamma_j^{-2} T_h(\wt A) + \gamma_j^{-4} }
 - \gamma_j^{-2} - C h^2 2^{-2j} R^{-2} \Big\} \Phi_{+}\phi_{j,R}\Big]_{-}\nonumber \\
&\geq
- C- Ch^{-1} \int |\nabla \otimes \wt A|^2.
\end{align}
So in conclusion, using the choice $R=h^{-1}$,
\begin{align}
  \label{eq:3}
\tr\Big[&\Phi_{+}\Big\{ \sqrt{\beta^{-2} T_h(\wt A) + \beta^{-4} - C h^2
  \beta^{-2} R^{-2} {\bf 1}_{\{ d(x) \leq 2R\}}} - \beta^{-2} -\wt\kappa V^{TF}_{{\mathbf z}, {\mathbf r}}
\Big\} \Phi_{+}\Big]_{-} \nonumber \\
  &\geq
- C- Ch^{-1} \int |\nabla \otimes \wt A|^2.
\end{align}

\subsection{The semiclassical region}

In this section we estimate  
 the intermediate region, the second term in \eqref{eq:5}. 
We apply a multiscale analysis by localizing the intermediate regime
into balls of varying radii such that each radius is comparable with
the distance of the center of the ball to the nearest nucleus. 
We then rescale the problem in each ball to a model problem
in the unit ball with new parameters $h$ and $\beta$.
The model problem is analyzed in Section~\ref{sec:locsc},
here we state the scaled version of the main result: 

\begin{theorem}[Scaled local semiclassics]\label{thm:SemiclassLocScaled}
Let $\ell, f, \lambda >0$.
Let $\theta$ be a bounded smooth cutoff function supported on the ball $B(\ell)$ 
and let $V$ a smooth real potential on $B(\ell)$. 
Assume that there is a constant $C'$ and
 for any multiindex $n\in {\mathbb N}^3$ there is a constant $C_n$ such that
$$
|\ell^{|n|} \partial^n \theta|+
|f^{-2} \ell^{|n|} \partial^n V| \leq C_n,  \;\; \mbox{and} \quad \beta f^2 \ell \leq C'h.
$$
Then
\begin{align}
\Big| &\inf_A \Big\{ 
\tr\Big[ \theta \big\{ \sqrt{\beta^{-2} T_h(A) + \beta^{-4}} - \beta^{-2} - V \big\} \theta \Big]_{-}
+ \frac{\lambda f \ell^2 }{\beta^2 h^3} \int_{B(2\ell)}|\nabla \otimes A|^2 \Big\} \nonumber \\
&- \frac{2}{(2\pi h)^3} \iint \theta(x)^2 \Big[\frac{1}{2} p^2-V(x)\Big]_{-} dx dp \Big|
\leq C h^{-2+1/11} \ell^{2-1/11} f^{4-1/11},
\end{align}
where $C$ depends on $\lambda$, $C'$ and on finitely many constants $C_n$.
\end{theorem}

Theorem~\ref{thm:SemiclassLocScaled} follows from the unscaled version
 Theorem~\ref{thm:SemiclassLoc} below  with the rescaled variable
$x'=x/\ell$ and using the parameters $\beta'= \beta f$,
 $h' = h/(f\ell)$. $\;\;\Box$

\medskip

The multiscale analysis requires two scaling functions, $\ell(u)=\ell_u$ and $f(u)=f_u$
depending on $u\in {\mathbb R}^3$. They express the lengthscale and the
size of the potential around $u$, respectively. In our case
we define 
$$
\ell(u)=\ell_u: = \frac{1}{100}\sqrt{ r^2+ d(u)^2}, \qquad 
f(u)=f_u:= \min\{\ell_u^{-1/2}, \ell_u^{-2}\},
$$
where we recall the definition of $d(u)=d_{{\mathbf r}}(u)$ from \eqref{ddef}
and that $r=h^{3/2}$.
The function $\ell_u$ is essentially the distance from $u$ to the nearest nucleus,
regularized on scale $r$, i.e. $\ell_u$ and $d(u)$ are comparable if $d(u)\ge r/3$.
The scaling function $f^2_u$ is the size of the Thomas-Fermi potential $V^{TF}_{{\mathbf z}, {\mathbf r}}$
near the point $u$.
More precisely, it follows from \eqref{VTFbound} that
\begin{align}\label{VTFs}
      \big|\partial_x^n V^{TF}_{{\mathbf z}, {\mathbf r}}(x)\big| \le C_n
   f_u^2 \ell_u^{-|n|}, \quad n\in {\mathbb N}^3,
\end{align}
for any $x$ with $|x-u|\le \ell_u$ and $d(u)\ge r/3$.
Moreover,  $\ell(u)$ is a continuously differentiable function with $\|\nabla\ell\|_\infty<1$.

Fix a cutoff function $\theta\in C_0^\infty({\mathbb R}^3)$, $0\le \theta\le 1$,
 supported in the unit ball and satisfying $\int \theta^2=1$.
Define
$$
   \theta_u(x) : = \theta\Big(\frac{x-u}{\ell(u)}\Big)\sqrt{J(x,u)}\ell(u)^{3/2},
$$
where $J(x,u)$ is the Jacobian of the (invertible) map $u\to (x-u)/\ell(u)$.
Then Theorem 22 from \cite{SS} states that 
\begin{align}\label{partnorm}
     \int_{{\mathbb R}^3} \theta_u(x)^2\ell_u^{-3}du =1
\end{align}
 for any $x\in {\mathbb R}^3$ and $\| \partial^n \theta_u\|_\infty\le C_n\ell_u^{-|n|}$
for any multiindex $n$.

Inserting this partition of unity and reallocating the localization error, we have
$$
 \beta^{-2} T_h(\wt A)- C\beta^{-2}h^2W_{r,R}+\beta^{-4}
\ge  \int \theta_u \Big[ \beta^{-2}T_h(\wt A)+\beta^{-4}-C\beta^{-2}h^2\ell_u^{-2}\Big]\theta_u \; du,
$$
where we also used that $W_{r,R}\le C\ell_u^{-2}$.
Since $\beta\le C'h$ and
$\ell_u \ge r/100 = h^{3/2}/100$, we have $C\beta^{-2}h^2\ell_u^{-2} \le C \beta^{-2}h^{-1}
\ll\beta^{-4}$ and we can thus define
$$
   \wt \beta_u := \beta(1-Ch^2\beta^2\ell_u^{-2})^{-1/4} = \beta \big[1+ O\big(h^2\beta^2\ell_u^{-2}\big)\big].
$$
Using the monotonicity of the square root and the pull-out estimate
of Lemma~\ref{lem:Pull-out}, we get
\begin{align}
\sqrt{ \beta^{-2} T_h(\wt A)- C\beta^{-2}h^2W_{r,R}+\beta^{-4}}-\beta^{-2}
  \ge \int \theta_u \Big[ \sqrt{ \wt \beta_u^{-2} T_h(\wt A)+\wt\beta_u^{-4}}-\wt\beta_u^{-2} -Ch^2\ell_u^{-2}
 \Big]\theta_u.
\label{puu}
\end{align}
(Strictly speaking, the pull-out estimate was formulated for a countable partition of unity, but the
integration over $u$ can be approximated by a discrete sum up to arbitrary precision, and
we neglect this technicality.)
For any potential $U\ge0$ set
$$
  \cE(\wt A, U,\theta_u): = \tr\Bigg[\Phi_-\phi_+\theta_u
 \Big(\sqrt{ \wt\beta_u^{-2} T_h(\wt A)+\wt\beta_u^{-4}}-\wt\beta_u^{-2}
  - U \Big)\theta_u\phi_+\Phi_{-} \Bigg]_- +
 \frac{c_1\wt\lambda}{\beta^2h^3}\int_{B_u(2\ell_u)} |\nabla\otimes \wt A|^2
$$
with a sufficiently small universal constant $c_1$.
Define the region
\begin{align}
\label{def:Q}
   \cQ: = \{ u\; : \; |u|\le 2R, \; |u-r_k|\ge r/3,\; k=1,2, \ldots, M\}
\end{align}
which supports $\Phi_-\phi_+$. It is easy to check that  $\theta_u\Phi_-\phi_+=0$ for $u\in \cQ^c$,
in particular $ \cE(\wt A, U,\theta_u)\ge 0$ in this case.

Using \eqref{puu} and reallocating the field energy (this is why $c_1$ needs to be small) we 
obtain
\begin{align}
\tr\Bigg[\Phi_-\phi_+& \Big(\sqrt{ \beta^{-2} T_h(\wt A)- C\beta^{-2}h^2W_{r,R}+\beta^{-4}}-\beta^{-2}
  -  \wt  \kappa V^{TF}_{{\mathbf z}, {\mathbf r}} \Big)\phi_+\Phi_{-} \Bigg]_- +
 \frac{\wt\lambda}{\beta^2h^3}\int |\nabla\otimes \wt A|^2\nonumber\\
& \ge \int_\cQ \frac{du}{\ell_u^3}\cE(\wt A, V_u^+,\theta_u),
\end{align}
where
$$
     V_u^+ : =  \wt  \kappa V^{TF}_{{\mathbf z}, {\mathbf r}} +Ch^2\ell_u^{-2}.
$$
For $u\in \cQ$
the localization error $C h^2 \ell_u^{-2}$ can be bounded by $Cf_u^2$ since $h\le C\ell_uf_u$
holds as long as  $ch^{2}\le\ell_u\le Ch^{-1}$ and $\ell_u$ is comparable with $u$.
Using $\wt\kappa\le 2M/\pi$ and \eqref{VTFs}, we see that
$$
   |\partial^n_x V_u^+(x)| \le C_n f_u^2\ell_u^{-|n|}
 \quad \mbox{for any} \; x\in \mbox{supp}(\theta_u), \quad u\in \cQ, \quad n\in {\mathbb N}^3.
$$
One can also easily check that
\begin{align}\label{thet}
    |\partial^n(\theta_u\Phi_-\phi_+)|\le C_n\ell_u^{-|n|}, \quad u\in \cQ.
\end{align} 

Noticing that $f_u\ell_u^2\le 1$ and that the condition $\wt\beta_u f_u\ell_u^2\le C' h$ is satisfied
by $\beta\le C'h$ (upon changing the value of $C'$),
 we can now apply Theorem~\ref{thm:SemiclassLocScaled}
to evaluate $\cE(\wt A, V_u^+,\theta_u)$:
\begin{align}
\int_\cQ &\frac{du}{\ell_u^3}  \cE(\wt A, V_u^+,\theta_u) \label{cEEE}\\
& \ge\int_\cQ \frac{du}{\ell_u^3} \Big[
 \frac{2}{(2\pi h)^3} \iint [(\theta_u\Phi_-\phi_+)(x)]^2 
\Big[\frac{1}{2} p^2-V_u^+(x)\Big]_{-} dx dp 
- C h^{-2+1/11} \ell_u^{2-1/11} f_u^{4-1/11}\Big].
\nonumber
\end{align}
The second term is of order $h^{-2+1/44}$, hence negligibile, since
$$
   \int_\cQ \frac{du}{\ell_u^3}\ell_u^{2-1/11} f_u^{4-1/11}
  \le \int_\cQ \min\{ \ell_u^{-3-1/22}, \ell_u^{-9+1/11}\}du \le C  r^{-1/22}
\le C h^{-3/44},
$$
using that $\ell_u\sim d(u)\ge r/3= h^{-3/2}/3$.

The double $dxdp$ integral in the leading term is of order
$$
 \int [(\theta_u\Phi_-\phi_+)(x)]^2 \big[ V_u^+(x)\big]_+^{5/2} dx
 = \int [(\theta_u\Phi_-\phi_+)(x)]^2 
\big[ \wt  \kappa V^{TF}_{{\mathbf z}, {\mathbf r}} +Ch^2\ell_u^{-2}\big]_+^{5/2} dx,
$$
which can be bounded by
\begin{align}
    (1+\epsilon)\int [(\theta_u\Phi_-\phi_+)(x)]^2
 \big[  \wt  \kappa V^{TF}_{{\mathbf z}, {\mathbf r}}\big]_+^{5/2} dx
   + C\epsilon^{-3/2}\int_{|x|\le CR} \Big[\frac{h^2}{ d(x)+r}\Big]^{5/2} dx
\nonumber\\
\label{eee}
\le \int [(\theta_u\Phi_-\phi_+)(x)]^2
 \big[  \wt  \kappa V^{TF}_{{\mathbf z}, {\mathbf r}}\big]_+^{5/2} dx
 + C\epsilon + C\epsilon^{-3/2}h^5R^{1/2} 
\end{align}
since $ V^{TF}_{{\mathbf z}, {\mathbf r}} \in L^{5/2} ({\mathbb R}^3)$
and $\ell_u$ is comparable with $d(x)+r$ for $x\in\mbox{supp} \; \theta_u$.
Optimizing for $\epsilon = h^2 R^{1/5} =h^{2-1/5}$, we see that the two
error terms in \eqref{eee}, after multiplying them with $h^{-3}$ 
and integrating over $\int_{\cQ} \ell_u^{-3} du$, are
of order $h^{-1-1/5}|\log h|$, i.e. negligible.

In \eqref{cEEE} we can thus replace $V_u^+$ with $ \wt  \kappa V^{TF}_{{\mathbf z}, {\mathbf r}}$
modulo irrelevant errors and we get
\begin{align} \nonumber
\int_\cQ \frac{du}{\ell_u^3}  \cE(\wt A, V_u^+,\theta_u) 
 & \ge\int_\cQ \frac{du}{\ell_u^3} 
 \frac{2}{(2\pi h)^3} \iint [(\theta_u\Phi_-\phi_+)(x)]^2 \Big[\frac{1}{2} p^2- 
\wt  \kappa V^{TF}_{{\mathbf z}, {\mathbf r}} (x)\Big]_{-} dx dp  -  Ch^{-2+1/44}
\\ 
& \ge  
 \frac{2}{(2\pi h)^3} \iint [(\Phi_-\phi_+)(x)]^2 \Big[\frac{1}{2} p^2- 
\wt  \kappa V^{TF}_{{\mathbf z}, {\mathbf r}} (x)\Big]_{-} dx dp  -  Ch^{-2+1/44}
\nonumber
\end{align}
after extending the $du$ integration to ${\mathbb R}^3$
and then performing it by using \eqref{partnorm}.

Finally, we can remove the cutoff function $\Phi_-$ in the last term by computing
that
\begin{align}
  \Bigg| \frac{2}{(2\pi h)^3} \iint &  (1-\Phi_-(x)^2)\phi_+(x)^2 \Big[\frac{1}{2} p^2- 
\wt  \kappa V^{TF}_{{\mathbf z}, {\mathbf r}} (x)\Big]_{-} dx dp\Bigg| \nonumber \\
&  \le Ch^{-3} \int \big[V^{TF}_{{\mathbf z}, {\mathbf r}}\big]^{5/2}  {\bf 1}(|x|\ge R)
\nonumber\\
& \le Ch^{-3} R^{-7} = Ch^4,
\nonumber
\end{align}
which is negligible.

Thus the final result for the second term in \eqref{eq:5} is
that for any $\wt\lambda>0$ and any admissible $\wt A$ we have
\begin{align}\nonumber
\tr \Big[\Phi_{-} \phi_+ &\Big( \sqrt{\beta^{-2} T_{h}(\wt A)-
  C \beta^{-2} h^2 W_{r,R}+ \beta^{-4}} - \beta^{-2} - \wt
  \kappa V^{TF}_{{\mathbf z}, {\mathbf r}} \Big)\phi_+\Phi_{-}\Big]_-
  + \frac{\wt\lambda}{\beta^2h^3}\int |\nabla\otimes \wt A|^2\\
&\ge   \frac{2}{(2\pi h)^3} \iint \phi_+(x)^2
\Big[ \frac{1}{2}p^2 -  \wt \kappa V^{TF}_{{\mathbf z}, {\mathbf r}}(x)
 \Big]_{-} \,dx dp -  Ch^{-2+1/44}.
\label{fin5}
\end{align}

\subsection{The region near the nuclei}
Here we will consider the first term in \eqref{eq:5}.
We consider each of the finitely many summands individually. 
Without loss of generality, we may assume that $r_k=0$, so we study 
\begin{align}
&\tr \Big[  \theta_{-}(|x|/r) \Big( \sqrt{\beta^{-2} T_{h}(\wt A)-
  C \beta^{-2} h^2 r^{-2} + \beta^{-4}} - \beta^{-2} - \wt
  \kappa V^{TF}_{ {\mathbf z}, {\mathbf r}} \Big)  \theta_{-}(|x|/r) \Big]_{-} \nonumber \\
 &\qquad 
+h^{1/4} \frac{\lambda}{\beta^2 h^3} \int |\nabla \otimes \wt A|^2,
\end{align}
where we also borrowed a small fraction $h^{1/4}$ of the magnetic energy.

Define 
$$
\wt \beta :=\beta (1- C h^2 \beta^2 r^{-2})^{-1/4} \geq \beta.
$$
Then we have 
\begin{align}
0\leq \beta^{-2} - \wt \beta^{-2} \leq C h^2 r^{-2} \leq C h^2 r^{-1} V^{TF}_{ {\mathbf z}, {\mathbf r}}(x),
\end{align}
on $|x|\leq r\le r_0/2$. Here we used from \eqref{VTFbound1} that for a small positive constant $c$
(depending on ${\mathbf z}$), we have
$$
V^{TF}_{ {\mathbf z}, {\mathbf r}}(x) \ge \frac{z_k}{|x|} - C \ge \frac{c}{|x|}
$$
for $|x| \leq r\ll 1$.
So we have
\begin{align}
\tr& \Big[  \theta_{-}(|x|/r) \Big( \sqrt{\beta^{-2} T_{h}(\wt A)-
  C \beta^{-2} h^2 r^{-2} + \beta^{-4}} - \beta^{-2} - \wt
  \kappa V^{TF}_{ {\mathbf z}, {\mathbf r}} \Big)  \theta_{-}(|x|/r) \Big]_{-} \nonumber \\
  &\geq
  \tr \Big[ \theta_{-}(|x|/r) \Big( \sqrt{\wt \beta^{-2} T_{h}(\wt A)+ \wt \beta^{-4}} - \wt \beta^{-2} -
\frac{\mu}{|x|} \Big)  \theta_{-}(|x|/r) \Big]_{-},
\nonumber
\end{align}
with
$$
   \mu: = ( \wt
  \kappa + C h^2 r^{-1})(z_k+Cr),
$$
where we used 
 the estimate 
$$
   V^{TF}_{ {\mathbf z}, {\mathbf r}}(x)\le \frac{z_k}{|x|}+ C\le \frac{z_k+Cr}{|x|}
$$
for $|x|\le r$
from \eqref{VTFbound1}.
By scaling $x=h^2\mu^{-1} y$, this becomes
\begin{align}
  \tr &\Big[  \theta_{-}(|x|/r) \Big( \sqrt{\wt \beta^{-2} T_{h}(\wt A)+ \wt \beta^{-4}} - \wt \beta^{-2} -
 \frac{\mu}{|x|} \Big)  \theta_{-}(|x|/r) \Big]_{-} \nonumber \\
  &\geq  \frac{\mu^2}{h^2} 
    \tr \Big[  \theta_{-}(|y|/{\mathcal R}) \Big( \sqrt{\wt \alpha^{-2} T_{h=1}(\bar A)+ \wt \alpha^{-4}}
 - \wt \alpha^{-2} -\frac{1}{|y|} \Big)  \theta_{-}(|y|/{\mathcal R}) \Big]_{-}
\end{align}
with 
\begin{align}
\wt \alpha: =  \wt \beta h^{-1} \mu,
\quad {\mathcal R}:= \frac{r}{h^2} \mu,\qquad
\bar{A}(y): = \frac{h}{\mu} \wt A(h^2 y/\mu).
\end{align}
Notice that
$$
\mu \int |\nabla \otimes \bar A(y)|^2 \,dy = \int |\nabla \otimes \wt A(x)|^2 \,dx,
$$
therefore 
\begin{align}
\frac{\lambda h^{1/4}}{\beta ^2 h^3}  \int |\nabla \otimes \wt A(x)|^2 \,dx
= \frac{\mu^2}{h^2} \Big\{ \Lambda \int |\nabla \otimes \bar A(y)|^2 \,dy\Big\},
\end{align}
with
$$
\Lambda = \frac{\lambda h^{1/4}}{\beta^2 h \mu}\gg 1.
$$
So with the choice $r=h^{3/2}$ we have (using $\beta \leq h$
and $\wt \kappa z_k\le C$),
\begin{align}
{\mathcal R}^5/\Lambda \leq C h^{1/4},
\end{align}
where the constant depends on $\lambda$.
With this choice we also have
$$
  \wt\alpha = \wt \beta h^{-1} \mu= \beta h^{-1}( \wt
  \kappa + C h^2 r^{-1})(z_k+Cr)  (1 - C h^2 \beta^2 r^{-2})^{-1/4}
   = \beta h^{-1} \wt \kappa z_k + O(h^{1/2}).
$$
So using Lemma~\ref{lem:ScottFunction} below and the continuity of $S_2$, we get,
after rescaling to the original coordinates that
\begin{align}
&\liminf_{h\rightarrow 0} h^2 \Bigg\{\tr \Big[  \theta_{-}(|x|/r) \Big( \sqrt{\beta^{-2} T_{h}(\wt A)-
  C \beta^{-2} h^2 r^{-2} + \beta^{-4}} - \beta^{-2} - 
\wt  \kappa V^{TF}_{ {\mathbf z}, {\mathbf r}} \Big)  \theta_{-}(|x|/r) \Big]_{-} \nonumber \\
&\qquad \qquad\qquad
+h^{1/4} \frac{\lambda}{\beta^2 h^3} \int |\nabla \otimes \wt A|^2
-\frac{2}{(2\pi h)^3} \iint\theta_-^2(|x|/r) \Big[\frac{1}{2} p^2 - \frac{\mu}{|x|}
\Big]_{-} dxdp \nonumber\\
&\qquad \qquad\qquad\label{fin61}
 - 2 (\wt \kappa z_k)^2 S_2(\beta h^{-1} \wt \kappa z_k )\Bigg\}\ge 0.
\end{align}
We can replace $\mu/|x|$ with $\wt  \kappa V^{TF}_{ {\mathbf z}, {\mathbf r}}(x)$ in the semiclassical formula,
the error is of order 
$$
  \Bigg|\int\theta_-^2(|x|/r) \Bigg\{ \Big[\frac{\mu}{|x|}\Big]^{5/2} - 
\big[ \wt  \kappa V^{TF}_{ {\mathbf z}, {\mathbf r}}(x)\big]^{5/2} \Bigg\} dx\Bigg|
\le C r^{1/2}h^{1/2} + C r^{3/2} \le Ch^{5/4},
$$
which is  negligible,
where we used \eqref{VTFbound1} and that $\mu = \wt\kappa z_k + O(h^{1/2})$.
After this replacement, we can sum up \eqref{fin61} for each $k$
to obtain the
final result of this section:

\begin{align}
&\liminf_{h\rightarrow 0} h^2 \sum_{k=1}^M  \Bigg\{ \tr \Big[  \theta_{r,k}   \Big( \sqrt{\beta^{-2} T_{h}(\wt A)-
  C \beta^{-2} h^2 r^{-2} + \beta^{-4}} - \beta^{-2} - 
\wt  \kappa V^{TF}_{ {\mathbf z}, {\mathbf r}} \Big)  \theta_{r,k} \Big]_{-} \nonumber \\
&\qquad \qquad\qquad\label{fin6}
+h^{1/4} \frac{\lambda }{\beta^2 h^3} \int |\nabla \otimes \wt A|^2
-\frac{2}{(2\pi h)^3} \iint\theta_{r,k}^2(x) \Big[\frac{1}{2} p^2 - \wt  \kappa V^{TF}_{ {\mathbf z}, {\mathbf r}}
\Big]_{-} dxdp \nonumber\\
& \qquad \qquad\qquad
 -  2 (\wt \kappa z_k)^2 S_2(\beta h^{-1} \wt \kappa z_k)\Bigg\} \ge 0.
\end{align}

\bigskip

Combining the estimates \eqref{eq:3}, \eqref{fin5} and \eqref{fin6}
on the three terms in \eqref{eq:5}
and recalling  \eqref{phimin},
 we immediately obtain 
\eqref{eq:scottcorr}. This completes the proof of
Theorem~\ref{thm:Thm1.4bis}.  $\;\;\Box$

\section{The Scott region}\label{sec:Scott}

In this section we fix a non-negative cutoff function $\phi: {\mathbb R}^3 \rightarrow [0,1]$
 with support on the unit ball $B(1)$ and such that $\phi \equiv 1$ on $B(1/2)$, the ball of
 radius $1/2$. Set $\phi_r(x) :=\phi(x/r)$ for any $r>0$.
Define, for $R, \Lambda>0$ and $\alpha \in (0,2/\pi)$, 
\begin{align}
{\mathcal E}_{R,\alpha, \Lambda}(A) &= \tr\Big[ \phi_{R}\big( \sqrt{ \alpha^{-2} T_{h=1}(A) + \alpha^{-4}}- \alpha^{-2}
- \frac{1}{|x|}   \big)\phi_{R}\Big]_{-}
+ \Lambda \int |\nabla \otimes A|^2\nonumber \\
&\quad-\frac{2}{(2\pi)^3} \iint\phi_{ R}^2(x) \Big[\frac{1}{2} p^2 - \frac{1}{|x|}\Big]_{-} dxdp
\end{align}
and 
\begin{align}
E(R, \alpha, \Lambda)
=
\inf_A {\mathcal E}_{R,\alpha,\Lambda}(A).
\end{align}

Clearly,
$E(R,\alpha,\Lambda) \leq   {\mathcal E}_{R,\alpha,\Lambda}(A=0)$, 
and we know from \cite[Lemma~4.3]{SSS} that
 $ {\mathcal E}_{R,\alpha,\Lambda}(A=0)$ tends to the non-magnetic, 
relativistic Scott term $2 S_2(\alpha)$ (the factor $2$ being due to the spin degrees of freedom).

\begin{lemma}\label{lem:ScottFunction}
Fix $\alpha_0 \in (0,2/\pi)$. We take the limits $R, \Lambda \rightarrow \infty$ in such a way that
 $R^5 /\Lambda \rightarrow 0$. Then we have
\begin{align}\label{simlim}
\lim_{R, \Lambda \rightarrow \infty, R^5/\Lambda \rightarrow 0} E(R, \alpha, \Lambda) = 2 S_2(\alpha),
\end{align}
where the limit is uniform in $\alpha \leq \alpha_0$.
\end{lemma}

\begin{proof}
As mentioned above, the upper bound follows by taking $A=0$ and using \cite[Lemma~4.3]{SSS}.
We proceed to give the lower bound.

\noindent
{\bf Step 1: A priori bound on the field energy.}
Theorem~\ref{thm:CritStability} with $V=0$ yields 
\begin{align}
\tr \Big[ \phi_{R} \big( \sqrt{\alpha^{-2}T_{h=1}(A) + \alpha^{-4}} - \alpha^{-2} - 1/|x|\big)  \phi_{R} \Big]_{-} 
\geq -C \Big\{ \int |\nabla \otimes A|^2 + R^3 \Big\},
\end{align}
with a constant $C$ that only depends on  $\frac{2}{\pi} - \alpha>0$,
i.e. on the distance of $\alpha$ from its critical value $\frac{2}{\pi}$.
Notice that the Weyl term also satisfies a similar bound,
\begin{align}\label{eq:ClassicalTerm}
\Big| \iint \phi_{R}^2(x) \left[\frac{1}{2} p^2 -\frac{\kappa}{|x|} \right]_{-} \,dxdp\Big|
= C_{\phi} \kappa^{5/2} R^{1/2} ,
\end{align}
for some constant $C_{\phi}$ only depending on $\phi$.

Inserting these bounds in ${\mathcal E}_{R,\alpha,\Lambda}$ we get for any
 $A$ with ${\mathcal E}_{R,\alpha,\Lambda}(A) \leq {\mathcal E}_{R,\alpha,\Lambda}(A=0)$ 
 that
\begin{align}\label{compareSSS}
\Big( \Lambda - C \Big) \int |\nabla \otimes A|^2 -C (R^3+R^{1/2})
\leq {\mathcal E}_{R,\alpha,\Lambda}(A=0).
\end{align}
We know from \cite{SSS} that ${\mathcal E}_{R,\alpha,\Lambda}(A=0)$
tends to $2S_2(\alpha)$ and $S_2$ is bounded by $1/4$.
In particular, the right hand side of \eqref{compareSSS}
is bounded by some constant $C$ for large values of $R$.
So we get for all $\Lambda, R$ sufficiently large
that
\begin{align}\label{eq:MagnApriori}
\int |\nabla \otimes A|^2 \leq C R^3/\Lambda.
\end{align}

\noindent
{\bf Step 2: Localization of the vector potential.} 
We start by localizing the vector potential $A$. 
Suppose that $A$ satisfies that ${\mathcal E}_{R,\alpha,\Lambda}(A) \leq {\mathcal E}_{R,\alpha,\Lambda}(A=0)$.
We may add a constant to $A$ without changing ${\mathcal E}_{R,\alpha,\Lambda}(A)$.
So we will assume that
$$
\int_{B(2R)} A \, dx = 0.
$$
Let $\chi_1, \chi_2 \in C^{\infty}({\mathbb R}^3)$ be a partition of unity satisfying
\begin{align}
\chi_1^2 + \chi_2^2 =1,\qquad \chi_1 = 1 \text{ on } B(1),\qquad \supp \chi_1 \subset B(2).
\end{align}
Define $\chi_{j, R}(x) = \chi_j(x/R)$, $j=1,2$.
Let also $\tilde \chi_1 \in C_0^{\infty}(B(2))$ with $\tilde \chi_1 = 1$ on $\supp \chi_1$.
We define 
$$
\tilde A(x) = \tilde \chi_1(x/R) A(x).
$$
With this notation we get from the IMS-formula (and since $\wt \chi_1 \chi_1 = \chi_1$)
\begin{align}
\alpha^{-2} T_{h=1}(A) + \alpha^{-4} 
&\geq \chi_{1, R} \big[ \alpha^{-2} T_{h=1}(\tilde A) -C \alpha^{-2} R^{-2} + \alpha^{-4} \big] \chi_{1, R}\nonumber \\
&\quad+
\chi_{2, R} \big[ \alpha^{-2} T_{h=1}(A) -C \alpha^{-2} R^{-2} + \alpha^{-4} \big] \chi_{2, R}.
\end{align}
Using the operator monotonicity of the square root, the pull-out estimate of Lemma~\ref{lem:Pull-out}
and that $\phi_R\chi_{1,R}=\phi_R$, $\phi_R\chi_{2,R}=0$,
we therefore have 
\begin{align}
\phi_{R} \sqrt{\alpha^{-2}T_{h=1}(A) + \alpha^{-4}} \phi_{R} 
&\geq
\phi_{R} \sqrt{\alpha^{-2}T_{h=1}(\tilde A) -C \alpha^{-2}R^{-2} + \alpha^{-4}} \phi_{R}\nonumber\\
&\geq  \phi_{R} \sqrt{\gamma^{-2}T_{h=1}(\tilde A) + \gamma^{-4}} \phi_{R},
\end{align}
where
\begin{align}
  \label{eq:1}
  \gamma = \alpha (1 - C\alpha^2 R^{-2})^{-1/4} \geq \alpha.
\end{align}
 
\noindent
{\bf Step 3: Removing the magnetic field.}
To continue the lower bound, we estimate
\begin{align}\label{eq:lowerDelta}
T_{h=1}(\tilde A) 
\geq -(1-2\epsilon) \Delta + \epsilon\Big( -\Delta - \epsilon^{-2} \tilde A^2\Big)
\end{align}
with some $\epsilon\in (0,1/2)$ to be determined later.
By the CLR-inequality $${-\Delta - \epsilon^{-2} \tilde A^2\geq 0}$$ if 
\begin{align}\label{eq:CondPos}
C_{\rm CLR} \int (\epsilon^{-2} \tilde A^2)^{3/2}  < 1,
\end{align}
where $C_{\rm CLR}$ is an explicit constant in the CLR inequality.
By the Cauchy-Schwarz and Sobolev inequalities,
and using that $\tilde A$ is supported on $B(2R)$, we obtain
\begin{align}\label{eq:EST-A3}
\int (\epsilon^{-2} \tilde A^2)^{3/2} & = \epsilon^{-3} 
 \Big(\int_{B(2R )} 1 \Big)^{1/2}
\Big(\int \tilde A^6 \Big)^{1/2}  \leq C \epsilon^{-3} R^{3/2}  \Big( \int |\nabla \otimes \tilde A|^2 \Big)^{3/2}.
\end{align}
We can continue the estimates using the Poincare inequality (since $\int_{B(2R)} A\,dx =0$).
\begin{align*}
\int |\nabla \otimes \tilde A|^2& \leq \int  |\nabla \otimes A|^2 + |\nabla \wt \chi_1(\cdot/R)|^2 A^2
\leq \int  |\nabla \otimes A|^2 + C  R^{-2} \int_{B(2R)} A^2 \nonumber \\
&\leq C' \int |\nabla \otimes A|^2.
\end{align*}
So we may replace $\tilde A$ by $A$ in \eqref{eq:EST-A3}.

Upon inserting this estimate in \eqref{eq:EST-A3} and using \eqref{eq:MagnApriori}, 
we see that the condition \eqref{eq:CondPos}
is satisfied if we take
\begin{align}\label{eq:ChoiceEpsilon}
\epsilon = \mu R^2 \Lambda^{-1/2}
\end{align}
with a sufficiently large constant $\mu$. Clearly $\epsilon\in (0,1/2)$
can be achieved in the limit
considered in \eqref{simlim}.

With the choice of $\epsilon$ from \eqref{eq:ChoiceEpsilon} and using \eqref{eq:lowerDelta} and the
 operator monotonicity of the square root, we have
\begin{align}\label{eq:NonMagnScott}
{\mathcal E}_{R,\alpha,\Lambda}(A) 
&\geq
\tr \Big[ \phi_{R} \big( \sqrt{\gamma^{-2}(1-2\epsilon)(-\Delta) + \gamma^{-4}}
 - \alpha^{-2}  - 1/|x|\big)  \phi_{R} \Big]_{-} \nonumber \\
&\quad -\frac{2}{(2\pi)^3}\iint \phi_{R}^2(x) \left[\frac{1}{2} p^2 -\frac{1}{|x|} \right]_{-} \,dxdp.
\end{align}
We perform the scaling $y = (1-2\epsilon)^{-1/2} x$ in order to absorb the 
factor $(1-2\epsilon)$. With the new parameter
$$
\widetilde R = (1-2\epsilon)^{-1/2} R,
$$
we get
\begin{align} 
&\tr \Big[ \phi_{R} \big( \sqrt{\gamma^{-2}(1-2\epsilon)(-\Delta) +
 \gamma^{-4}} - \alpha^{-2}  - 1/|x|\big)  \phi_{R} \Big]_{-} 
\nonumber \\
&=
\tr \Big[ \phi_{\wt R } \Big( \sqrt{\gamma^{-2}(-\Delta) +
  \gamma^{-4}} - \gamma^{-2} -(\alpha^{-2}-\gamma^{-2})
  - \frac{\sqrt{1-2\epsilon}}{|x|}\Big)  \phi_{\wt R} \Big]_{-} 
\end{align}
We use the Lieb-Thirring inequality Theorem~\ref{thm:LT} (in this case Theorem~\ref{thm:LT} is the usual Daubechies inequality) to control
the difference $(\alpha^{-2}-\gamma^{-2})$. 
For this we will use a small  $\delta$-part of 
the kinetic energy (in the end we will make the optimal choice $\delta = R^{-1}$).
Since 
$$
0 \leq \alpha^{-2}-\gamma^{-2} \leq C R^{-2},
$$
we get
\begin{align}
  \label{eq:2}
  \tr\Big[ &\phi_{\wt R }\big( \delta(  \sqrt{\gamma^{-2}(-\Delta) +
  \gamma^{-4}} - \gamma^{-2}) - (\alpha^{-2}-\gamma^{-2}) \big)
\phi_{\wt R } \Big]_{-}
\nonumber \\
&\geq
-C \int_{\{|x|\leq \wt R \}} \Big(\delta^{-3/2}(\alpha^{-2}-\gamma^{-2})^{5/2} +
\gamma^3 \delta^{-3}(\alpha^{-2}-\gamma^{-2})^4\Big) \nonumber \\
&\geq
-C (\delta^{-3/2}R^{-2} + \delta^{-3} R^{-5}).
\end{align}
With the choice $\delta = R^{-1}$ this term is estimated as $C R^{-1/2} $.

For the main term, containing $(1-\delta)$-part of the kinetic
energy and the Coulomb potential,
by scaling $x=\frac{1-\delta}{\sqrt{1-2\epsilon}} y$, we have
\begin{align}
(1-\delta)& \tr\Big[ \phi_{\wt R }\big(\sqrt{\gamma^{-2}(-\Delta) +
  \gamma^{-4}} - \gamma^{-2} - \frac{(1-\delta)^{-1}\sqrt{1-2\epsilon}}{|x|}\big)
\phi_{\wt R } \Big]_{-} \nonumber \\
&= \frac{1-2\epsilon}{1-\delta} \tr\Big[ \phi_{\bar R }
\big(\sqrt{\wt \gamma^{-2}(-\Delta) +
  \wt \gamma^{-4}} 
  - \wt \gamma^{-2} - \frac{1}{|x|}\big)
\phi_{\bar R } \Big]_{-},
\end{align}
with $\bar R = \wt R \frac{\sqrt{1-2\epsilon}}{1-\delta}$
and $\wt \gamma = \gamma \frac{\sqrt{1-2\epsilon}}{1-\delta}$.

Notice that the classical (Weyl) terms satisfy \eqref{eq:ClassicalTerm} and therefore,
\begin{align}
&\Big|\iint \phi_{R}^2(x) \left[\frac{1}{2} p^2 -\frac{1}{|x|} \right]_{-} \,dxdp
-
\frac{1-2\epsilon}{1-\delta}
\iint \phi_{\bar R}^2(x) \left[\frac{1}{2} p^2 -\frac{1}{|x|} \right]_{-} \,dxdp\Big|\nonumber \\
&=
C_{\phi} \Big| R^{1/2} - \frac{1-2\epsilon}{1-\delta} \bar R^{1/2}\Big|
\leq C  R^{1/2} [\epsilon+\delta] = o(1),
\end{align}
using the choice $\delta = R^{-1}$ and 
by the choice of $\epsilon$ and since $R^{1/2} \epsilon \rightarrow 0$ under the limit taken in \eqref{simlim}.

So we can insert the above estimates into \eqref{eq:NonMagnScott} to get
\begin{align}
{\mathcal E}_{R,\alpha,\Lambda}(A) &\geq \frac{1-2\epsilon}{1-\delta}{\mathcal E}_{\bar R,\wt \gamma, \Lambda}(A=0)
- C R^{1/2} [\epsilon+\delta]  - C (\delta^{-3/2}R^{-2} + \delta^{-3} R^{-5}).
\end{align}
Since we know from \cite{SSS} that ${\mathcal E}_{\bar R,\wt \gamma, \Lambda}(A=0) \rightarrow 2 S_2(\alpha)$
this finishes the proof using the previously mentioned choices $\delta = R^{-1}$ and $\epsilon$ from \eqref{eq:ChoiceEpsilon}.
\end{proof}

\section{Local semiclassics}\label{sec:locsc}

\begin{theorem}\label{thm:SemiclassLoc}
Let $\theta$ be a bounded cutoff function supported on the unit ball $B(1)$
and $V$ a smooth potential on $B(1)$. Let $\lambda>0$ be fixed.
Assume that there is a constant $C'$ and
for any $n\in {\mathbb N}^3$ there is a constant $C_n$ such that
$$
|\partial^n V| \leq C_n,  \;\; \mbox{and} \quad \beta \leq C'h.
$$

Then
\begin{align}
\Big| &\inf_A \Big\{ 
\tr\big[ \theta \big\{ \sqrt{\beta^{-2} T_h(A) + \beta^{-4}} - \beta^{-2} - V \big\} \theta \big]_{-}
+ \frac{\lambda }{\beta^2 h^3} \int_{B(2)}|\nabla \otimes A|^2 \Big\} \nonumber \\
&- \frac{2}{(2\pi h)^3} \iint \theta(x)^2 \Big[ \frac{1}{2} p^2-V(x)\Big]_{-} dx dp \Big|
\leq C h^{-2+1/11},
\end{align}
where $C$ depends on $\lambda$ and on finitely many constants $C_n$.
\end{theorem}

\begin{proof}
The upper bound follows by choosing $A=0$ and applying \cite[Theorem~4.1]{SSS}.
Also, using Theorem~\ref{thm:LT}, it suffices to prove the estimate when $h$ is sufficiently small,
 say $h < 10^{-11}$.

Let $\ell$ be a parameter satisfying $h\leq \ell \leq 1/10$. At the end of
 the proof we will choose $\ell = h^{1/11}$.
Let $\{ \phi_{j,\ell}\}_{j \in {\mathbb Z}^3}$ be a collection of smooth functions satisfying
\begin{align*}
\sum_j \phi_{j,\ell}^2 = 1, \qquad
\supp \phi_{j,\ell} \subset B_{j\ell} (\ell), \qquad
\sum_j |\nabla \phi_j|^2 \leq C \ell^{-2},
\end{align*}
where $B_x(r)$ is the ball of radius $r$ centered at $x$.
Then, by the IMS-formula
\begin{align*}
T_h(A) \geq \sum_j  \phi_{j,\ell} \big( T_h(A) - C h^2\ell^{-2} \big) \phi_{j,\ell} .
\end{align*}
So by the pull-out estimate of Lemma~\ref{lem:Pull-out} and operator monotonicity of the square root,
\begin{align}
\tr\Big[ &\theta \big\{ \sqrt{\beta^{-2} T_h(A) + \beta^{-4}} - \beta^{-2} - V \big\} \theta \Big]_{-} \\
&\geq \nonumber
\sum_j 
\tr\Big[ \theta \phi_{j,\ell} \big\{ \sqrt{\beta^{-2} T_h(A)
 -C \beta^{-2} h^2 \ell^{-2} + \beta^{-4}} - \beta^{-2} - V \big\} \phi_{j,\ell} \theta \Big]_{-}.
\end{align}
Also, with some universal constant $c_0$
$$
\sum_{j: B_{j\ell}(\ell) \cap B(1) \neq \emptyset} \int_{B_{j\ell}(2\ell)} |\nabla \otimes A|^2
\leq c_0  \int_{B(2)} |\nabla \otimes A|^2,
$$
so
\begin{align}\label{eq:AfterPullOut}
\tr&\Big[ \theta \big\{ \sqrt{\beta^{-2} T_h(A) + \beta^{-4}} - \beta^{-2} - V\big\} \theta \Big]_{-} 
+ \frac{\lambda}{\beta^2 h^3} \int_{B(2)} |\nabla \otimes A|^2
\nonumber \\
&\geq
\sum_j \Big\{ 
\tr\Big[ \theta \phi_{j,\ell} \big\{ \sqrt{\beta^{-2} T_h(A) -C \beta^{-2} h^2 \ell^{-2} + \beta^{-4}}
 - \beta^{-2} - V \big\} \phi_{j,\ell} \theta \Big]_{-}\nonumber \\
&\qquad+ \frac{\lambda}{c_0 \beta^2 h^3} \int_{B_{j\ell}(2\ell)} |\nabla \otimes A|^2\Big\}.
\end{align}
We may consider each summand independently and we can focus only 
on those that give negative contribution.
 For simplicity of notation, we take $j=0$. 
Choose a new, smooth partition of unity $\chi_1^2 + \chi_2^2=1$, with
$$
\supp \chi_1 \subset B(2\ell), \qquad \chi_1 = 1 \text{ on } \supp \phi_{0,\ell}, \qquad
|\nabla \chi_1|^2 + |\nabla \chi_2|^2 \leq C \ell^{-2}.
$$

By a constant shift in $A$ and gauge invariance we may assume that $\int_{B(2\ell)} A = 0$.
 Choose $\wt A = \wt \chi A$, with $\wt \chi$ satisfying the same conditions
 as $\chi_1$ and $\wt \chi \chi_1 = \chi_1$. 
Define $W_{h,\ell} = h^2 (|\nabla \chi_1|^2 + |\nabla \chi_2|^2)$.

Then, by IMS and the pull-out estimate again
\begin{align}
 \phi_{j,\ell} &\sqrt{\beta^{-2} T_h(A) -C \beta^{-2} h^2 \ell^{-2} + \beta^{-4}}  \phi_{j,\ell} 
\nonumber \\
& =
  \phi_{j,\ell} \sqrt{\beta^{-2} (\chi_1 T_h(\wt A) \chi_1+ \chi_2 T_h(A) \chi_2 - W_{h,\ell} ) -C \beta^{-2} h^2 \ell^{-2} + \beta^{-4}}  \phi_{j,\ell} \nonumber \\
&\geq
\phi_{j,\ell} \sqrt{\beta^{-2} T_h(\wt A) -C' \beta^{-2} h^2 \ell^{-2} + \beta^{-4}}  \phi_{j,\ell} 
\label{usepullout}
\end{align}
with a different constant $C'$, where we used that $\chi_1\phi_{j,\ell}=\phi_{j,\ell}$.

Using the Poincare inequality, we also have
\begin{align}
\int_{{\mathbb R}^3} |\nabla \otimes \wt A|^2 &= \int_{B(2\ell)} |\nabla \otimes \wt A|^2 \leq
 \int_{B(2\ell)} \wt \chi^2 |\nabla \otimes A|^2 + 2|\nabla \wt \chi|^2 A^2 \nonumber \\
& \leq C_1  \int_{B(2\ell)} |\nabla \otimes A|^2
\label{usepoincare}
\end{align}
for some universal constant $C_1$.

So for each $j$, it suffices to consider a semiclassical lower bound to
\begin{align}
\inf_A \, & \Bigg\{\tr\Big[ \theta \phi_{j,\ell} \big\{ \sqrt{\beta^{-2} T_h(A)
 -C \beta^{-2} h^2 \ell^{-2} + \beta^{-4}} - \beta^{-2} - V\big\} \phi_{j,\ell} \theta \Big]_{-} \nonumber \\
&\qquad +
 \frac{\lambda}{c_0 \beta^2 h^3} \int_{B_{j\ell}(2\ell)} |\nabla \otimes A|^2 \Bigg\},
\label{eachj}
\end{align}
where $c_0$ is a given fixed constant, and the infimum is taken
 over all vector fields $A$ with support contained in $B_{j\ell}(2\ell)$.

First we get a crude upper bound on the field energy.
Clearly
$$
  \sqrt{\beta^{-2} T_h(A) -C \beta^{-2} h^2 \ell^{-2} + \beta^{-4}} - \beta^{-2}
 \ge \sqrt{\wt\beta^{-2} T_h(A) + \wt\beta^{-4}} - \wt\beta^{-2} - Ch^2\ell^{-2}
$$
with  $\wt \beta$ defined by
$$
    \wt\beta^{-4} = \beta^{-4} -C \beta^{-2} h^2 \ell^{-2}.
$$
We can apply the Lieb-Thirring inequality Theorem~\ref{thm:LT}
for the first line of \eqref{eachj}
 with the bounded potential $V(x)+Ch^2\ell^{-2}{\bf 1}(x\in \mbox{supp } \phi_{j,\ell})$
to obtain a lower bound of order $- Ch^{-3}\ell^3$. This implies that
\begin{align}\label{eq:apriori}
{\mathcal B}^2 := \int_{B(2\ell)} |\nabla \otimes A|^2 \leq C\beta^2 \ell^3,
\end{align}
with a large constant $C$
whenever the vector field $A$ gives a non-positive energy.

\medskip

We now estimate, for any $\epsilon\in (0,1)$,
\begin{align}\label{eq:Diamagnetic2}
T_h(A) &\geq -(1-2\epsilon) h^2 \Delta + \epsilon( -h^2\Delta - \epsilon^{-2} A^2) \nonumber \\
&\geq -(1-2\epsilon) h^2 \Delta - C h^{-3} \epsilon^{-4} \ell^{1/2} {\mathcal B}^5.
\end{align}
Here the last inequality follows from the Lieb-Thirring, H\"{o}lder and Sobolev inequalities
recalling that $A$ is supported in $B_{j\ell}(2\ell)$:
\begin{align*}
 -h^2\Delta - \epsilon^{-2} A^2 &\geq -C h^{-3} \epsilon^{-5} \int A^5 \geq
 -C h^{-3} \epsilon^{-5} \Big(\int_{B(2\ell)}1 \Big)^{1/6} \Big(\int A^6 \Big)^{5/6} \\
& \geq - C  h^{-3} \epsilon^{-5} \ell^{1/2}\Big( \int |\nabla \otimes A|^2\Big)^{5/2}.
\end{align*}
Define 
\begin{align}
\gamma^{-4} := \beta^{-4} - C\beta^{-2} h^2 \ell^{-2} -
 C \beta^{-2} h^{-3} \epsilon^{-4} \ell^{1/2} {\mathcal B}^5,\qquad
\tilde{h} = h \sqrt{1-2\epsilon}.
\end{align}
We will in the end make the (optimal) choice
\begin{align}\label{eq:Choice}
\epsilon = h^{-3/5} \ell^{1/10} {\mathcal B}.
\end{align}
Using \eqref{eq:apriori}, this choice will ensure that 
$$
\epsilon \ll 1, \qquad 
h^{-3} \epsilon^{-4}  \ell^{1/2} {\mathcal B}^5 =  h^{-3/5} \ell^{1/10} {\mathcal B} \ll 1,
$$
so $\gamma$ is well defined.
Using operator monotonicity of the square root  we get
from \eqref{eq:Diamagnetic2} that
\begin{align}\label{diam3}
&\sqrt{\beta^{-2} T_h(A) + \beta^{-4}} - \beta^{-2} - V(x)  \\
&\quad\geq
\sqrt{\gamma^{-2} (-\tilde{h}^2 \Delta) + \gamma^{-4}} - \beta^{-2} - V(x) \nonumber\\
&\quad\geq
\sqrt{\gamma^{-2} (-\tilde{h}^2 \Delta) + \gamma^{-4}} - \gamma^{-2} - (V(x)+ C  h^{-3} \epsilon^{-4} \ell^{1/2}{\mathcal B}^5 + C h^2 \ell^{-2}). \nonumber 
\end{align}
Using \cite[Theorem~4.1]{SSS} we therefore have
\begin{align}
&\tr\Big[ \theta \phi_{0,\ell} \big\{ \sqrt{\beta^{-2} T_h(A) + \beta^{-4}} - \beta^{-2}
 - V(x) \big\}  \phi_{0,\ell} \theta \Big]_{-}
+ \frac{\lambda}{c_0\beta^2 h^3} \int_{B(2\ell)}|\nabla \otimes A|^2 \nonumber \\
&\quad\geq
 \frac{2}{(2\pi h)^3} \iint \theta(x)^2  \phi_{0,\ell}^2(x)\Big[
 \sqrt{\beta^{-2} p^2 + \beta^{-4}} - \beta^{-2}-V(x)\Big]_{-} dx dp - C(h/\ell)^{-9/5} \nonumber \\
 &\qquad
 - C (h/\ell)^{-3} \big( \epsilon + h^{-3} \epsilon^{-4} \ell^{1/2} {\mathcal B}^5 \big)
 + \beta^{-2} h^{-3} {\mathcal B}^2.
\end{align}
The leading semiclassical term is of order $(h/\ell)^3$.
The
$\epsilon$ term in the second line comes from adjusting $\wt h$ to $h$
in the main term and
we have absorbed the term $C (h/\ell)^{-3} h^2 \ell^{-2}$ 
(from the last line of \eqref{diam3}) in $(h/\ell)^{-9/5}$.
In the leading  term we can replace 
$ \sqrt{\beta^{-2} p^2 + \beta^{-4}} - \beta^{-2}$ with $\frac{1}{2}p^2$ at
the expense of a $\beta^2(h/\ell)^{-3}$ error.

By the choice of $\epsilon$ above, the last line is
$$
-Ch^{-3} \ell^{31/10} h^{-3/5} {\mathcal B} + \beta^{-2} h^{-3} {\mathcal B}^2.
$$
Clearly,
$$
Ch^{-3} \ell^{31/10} h^{-3/5} {\mathcal B} \leq \beta^{-2} h^{-3} {\mathcal B}^2 +
C^2h^{-3} \ell^{31/5} \beta^2 h^{-6/5}
\leq \beta^{-2} h^{-3} {\mathcal B}^2 + C^2h^{-3} (\ell^{31/5} h^{4/5}),
$$
using the bound $\beta \leq C'h$.
So we get
\begin{align}\label{eq:LocalSemiclassFinal}
&\tr\Big[ \theta  \phi_{0,\ell} \big\{ \sqrt{\beta^{-2} T_h(A) + \beta^{-4}} - \beta^{-2}
 - V(x) \big\}  \phi_{0,\ell} \theta \Big]_{-}
+ \frac{\lambda}{c_0 \beta^2 h^3} \int_{B(2\ell)}|\nabla \otimes A|^2 \nonumber \\
&\quad\geq
 \frac{2}{(2\pi h)^3} \iint \theta(x)^2  \phi_{0,\ell}^2(x)\Big[\frac{1}{2} p^2-V(x)\Big]_{-} dx dp 
- C(h/\ell)^{-9/5} \nonumber \\
 &\qquad
 - C (h/\ell)^{-3} (\beta^2+\ell^{16/5} h^{4/5}).
\end{align}
With the choice $\ell = h^{1/11}$, we have
$$
(h/\ell)^{-9/5} + (h/\ell)^{-3} (\ell^{16/5} h^{4/5}) = 2  (h/\ell)^{-3} h^{12/11}
$$
and the error term from $\beta^2$ is negligible since $\beta\le C'h$.

Similar bound holds for any $j$.  We proceed to insert \eqref{eq:LocalSemiclassFinal}
for each $j$ in \eqref{eq:AfterPullOut} and get
\begin{align}
\tr&\Big[ \theta \big\{ \sqrt{\beta^{-2} T_h(A) + \beta^{-4}} - \beta^{-2} - V(x) \big\} \theta \Big]_{-} 
+ \frac{\lambda}{c_0\beta^2 h^3} \int_{B(2)} |\nabla \otimes A|^2
\nonumber \\
&\geq  \frac{2}{(2\pi h)^3} \iint \theta(x)^2  \Big[\frac{1}{2} p^2-V(x)\Big]_{-} dx dp  - C h^{-3+12/11}.
\end{align}
Here we used that the summation in \eqref{eq:AfterPullOut} can be restricted to those $j$, 
where the ball $B_{j\ell}(\ell)$ has non-empty intersection with $B(1)$. By a volume argument 
there are of order of magnitude $\ell^{-3}$ such balls.

\end{proof}

\section{Proof of the relativistic Lieb-Thirring inequalities}\label{sec:proof}

\begin{proof}[Proof of Theorem~\ref{thm:LT}.]
By scaling it suffices to prove that there exists a constant $C>0$ such that for all $m\geq 0$,
\begin{align}\label{eq:noscale}
\tr\Big[& \sqrt{ T(A) + m^2} - m - V(x)\Big]_{-} \nonumber \\
&\geq - C \Big\{ m^{3/2} \int [V]_{+}^{5/2} + \int [V]_{+}^{4} + \Big( \int |\nabla\times A|^2 \Big)^{3/4}
 \Big( \int [V]_{+}^4\Big)^{1/4} \Big\} .
\end{align}
The basic idea is to consider the spectral subspaces on which $T=T(A)\le 10m^2$ and
its complement, $T\ge 10m^2$, separately.  For any $T\ge 0$ and $m\ge  0$ we have
the following simple arithmetic
inequalities:
\begin{align}
    \sqrt{T+m^2}- m \ge  & \; \frac{c_0T}{m}, \qquad \mbox{if $T< 10m^2$}\label{TT}\\
   \sqrt{T+m^2}- m \ge  & \; \frac{2}{3}\sqrt{T} , \qquad \mbox{if $T\ge 10m^2$} \nonumber
\end{align}
where $c_0$ is a universal constant (actually $c_0=\frac{1}{10}(\sqrt{11}-1)$ will do).
In the first regime we can use the non-relativistic magnetic Lieb-Thirring inequality.
In the second regime we will use the BKS inequality \cite{BKS} stating that
\begin{equation}\label{eq:BKS}
 \tr (P-Q)_-\ge - \tr[-(P^2-Q^2)_-]^{1/2} 
\end{equation}
for any positive operators $P,Q$. In this way we can turn the problem 
in the second regime into a Lieb-Thirring type estimate on the half moments of the
negative eigenvalues of the
Pauli operator.

For the detailed proof, we can clearly assume that $V \geq 0$. 
We split the potential as
$$
   V = V_1+V_2, \qquad V_1(x): = V(x)\cdot {\bf 1}_{\{V(x)\le m/2\}}, \qquad  V_2(x): = V(x)\cdot {\bf 1}_{\{V(x)> m/2\}}.
$$
Since
\begin{align*}
\tr\Big[ \sqrt{T + m^2} - m - V_1(x)- V_2(x)\Big]_{-}\geq &
\; \frac{1}{2}\tr\Big[ \sqrt{T + m^2} - m - 2V_1(x)\Big]_{-}\\
& +
\frac{1}{2}\tr\Big[ \sqrt{T + m^2} - m - 2V_2(x)\Big]_{-},
\end{align*}
for the proof of \eqref{eq:noscale}
it suffices to consider separately the cases 
\begin{itemize}
\item $V(x) \le m$ for all $x$;
\item $V(x) > m$ whenever $V(x)\ne 0$.
\end{itemize}
The first case will be applied to $V$ being $2V_1$, while the second to $V$ being $2V_2$.

In the first case, $V \leq m$, consider the projections
$$
P_{<}  := {\bf 1}_{\{ T < 10 m^2\}}, \qquad P_{\geq} := {\bf 1}_{\{ T \geq 10 m^2\}},
$$
and estimate
$$
V = (P_{<} + P_{\geq}) V (P_{<} + P_{\geq}) \leq 2 P_{<} V P_{<} + 2 P_{\geq} V P_{\geq}. 
$$
Since $V \leq m$, we have
\begin{align}\label{eq:SqrtToPauli}
\sqrt{T + m^2}- m - V 
&\geq
P_{<} \Big( \sqrt{T + m^2}- m - 2V \Big) P_{<} +
P_{\geq} \Big( \sqrt{T + m^2}- m - 2V \Big) P_{\geq} \nonumber \\
&\geq
P_{<} \Big( m^{-1} c_0 T - 2V \Big) P_{<},
\end{align}
where we used the spectral theorem and the elementary inequality \eqref{TT}.

Therefore,
\begin{align}
\tr\Big[ \sqrt{T+ m^2} - m - V(x)\Big]_{-} 
&\geq
m^{-1} \tr\big[ c_0 T - 2m V]_{-}  \\
&\geq
-C \Big\{ m^{3/2} \int [V]_{+}^{5/2}  + \Big( \int B^2 \Big)^{3/4} \Big( \int [V]_{+}^4\Big)^{1/4} \Big\},\nonumber
\end{align}
where the last inequality follows by the Lieb-Thirring inequality for the Pauli operator.

\medskip

We now consider the case where $V(x) \geq m$ whenever $V(x) \ne 0$. Notice, that in 
the special case $m=0$ this condition is automatically satisfied.
Here we first use the BKS inequality \eqref{eq:BKS} to get
\begin{align}
\tr\Big[ \sqrt{T + m^2} - m - V(x)\Big]_{-}
&\geq
-\tr\Big(-\Big[ T+ m^2 - (m+V)^2 \Big]_{-}\Big)^{1/2} \nonumber \\
&\geq
-\tr\Big(-\Big[ T - 3V^2 \Big]_{-}\Big)^{1/2}.
\end{align}
Here we estimated $m V \leq V^2$ by the assumption on $V$ to get the last inequality.

We now use the ``running energy scale'' method from \cite{LLS}. 
For a self-adjoint operator $H$ let ${\mathcal N}(H)$ denote the dimension of the negative spectral subspace.
Let $\lambda \in [0,1]$ be a real parameter chosen at the end.
Using that $T\ge 0$ and $T\ge (p+A)^2 -  |\nabla\times A| $, we obtain
\begin{align}
 \tr\Big(-\Big[ T - 3V^2 \Big]_{-}\Big)^{1/2} 
&=
\int_0^{\infty} {\mathcal N}\big( T - 3V^2 + e \big) \frac{de}{\sqrt{e}}\nonumber\\
&\leq
\int_0^{\infty} {\mathcal N}\big( \lambda T - 3V^2 + e \big) \frac{de}{\sqrt{e}} \nonumber \\
&\leq
\int_0^{\infty} {\mathcal N}\big( \lambda(p+A)^2 - \lambda |\nabla\times A| - 3V^2 + e \big) \frac{de}{\sqrt{e}} .
\end{align}
We now apply the CLR-estimate to get
\begin{align}
\tr\Big(-\Big[T - 3V^2 \Big]_{-}\Big)^{1/2}
&\leq
C \int_0^{\infty} \int_{{\mathbb R}^3} \big[ |\nabla\times A|
 + 3 \lambda^{-1} V^2 - \lambda^{-1} e\big]_{+}^{3/2} dx  \frac{de}{\sqrt{e}} \nonumber \\
&\leq
C'\Big\{ \sqrt{\lambda} \int |\nabla\times A|^2 + \lambda^{-2} \int V^4\Big\}.
\end{align}
Setting $B=|\nabla\times A|$ for simplicity,
if $\int B^2 \leq \int V^4$, we choose $\lambda =1$ and get a total estimate
 of size $\int V^4$. If $\int B^2 \leq \int V^4$ we choose
 $\lambda = (\int V^4/\int B^2)^{1/2}$ and get an estimate of size 
$ \Big( \int B^2 \Big)^{3/4} \Big( \int V^4\Big)^{1/4}$.
Thus the choice $\lambda = \min\big\{ 1, (\int V^4/\int B^2)^{1/2}\big\}$ will do the job
in both cases.
This finishes the proof of \eqref{eq:noscale} and therefore of Theorem~\ref{thm:LT}.
\end{proof}

\begin{proof}[Proof of Theorem~\ref{thm:CritStability}.] A potential with
 Coulomb singularity is not allowed in Theorem~\ref{thm:LT}. We will need to
use Kato's inequality to control the Coulomb singularity directly; the remaining part 
of the potential will be treated as in the proof of Theorem~\ref{thm:LT}. 
This will be done in  Lemma~\ref{thm:CritStability2}. However, this estimate
does not have the expected behavior for small values of $\beta$, where
one should be close to the non-relativistic situation. So  Lemma~\ref{thm:CritStability2}
 will be used only if $\beta$ is separated away
from zero, say $\beta\ge 1/20$. When $\beta \in (0, 1/20)$, we can
estimate $(\beta^{-2}T+\beta^{-4})^{1/2}-\beta^{-2}$ by $\mbox{(const.)} T$
effectively, and we use a result from \cite{ES2} to ``pull the Coulomb tooth''.
This is the content of Lemma~\ref{lem:smallbeta} below.

\begin{lemma}[Stability up to the critical coupling]\label{thm:CritStability2}
Let $r \geq r_0$ for some given $r_0>0$.
Let $\phi_r$ be a real function satisfying $\supp \phi_r \subset \{|x|\leq r\}$, $\| \phi_r \|_{\infty} \leq 1$.
There exists a constant $C>0$, depending only on $r_0$, such that if $\beta \in (0,2/\pi)$, then
\begin{align}\label{eq:stability}
&\tr \Big[ \phi_r \Big( \sqrt{\beta^{-2} T(A) + \beta^{-4}} - \beta^{-2} - \frac{1}{|x|}- V \Big) \phi_r \Big]_{-} \nonumber \\
&\geq
- C \Bigg\{ \eta^{-3/2} \beta^{-1}\int |\nabla \times A|^2 + \beta^{-5} \eta^{-3} r^3 + \eta^{-3/2} \int [V]_+^{5/2} 
+ \eta^{-3} \beta^3 \int [V]_+^4 \nonumber \\
&\qquad\qquad+ \Big(\int |\nabla \times A|^2 \Big)^{3/4} \Big[  \beta^{-1/2} r^{3/4} +
 \Big(\int [V]_+^4\Big)^{1/4} \Big] \Bigg\},
\end{align}
where $\eta = \frac{1}{10}(1-(\pi\beta/2)^2)$.
\end{lemma}

\begin{proof}[Proof of Lemma~\ref{thm:CritStability2}] Without loss of generality we can assume that $V\ge 0$.
We estimate
\begin{align}\label{eq:split}
\phi_r &\big[ \sqrt{\beta^{-2} T(A) + \beta^{-4}} - \beta^{-2} - \frac{1}{|x|}- V \big] \phi_r  \\
&\geq
\phi_r \big[ (1-\eta) \sqrt{\beta^{-2} T(A)} - \frac{1}{|x|} \big] \phi_r +
\phi_r \big[ \eta \sqrt{\beta^{-2} T(A) + \beta^{-4}} - \beta^{-2}{\bf 1}_{\{|x|\leq r\}} - V \big] \phi_r. \nonumber
\end{align}
In the first term we use the Kato inequality $(2/\pi)/|x| \leq |(-i\nabla +A)|$ and
 the BKS inequality \eqref{eq:BKS} to get
\begin{align}
\tr \Big[\phi_r \big( (1-\eta)& \sqrt{\beta^{-2} T(A)} - \frac{1}{|x|} \big) \phi_r\Big]_{-} \nonumber \\
&\geq
\tr\Big[ \phi_r \big( (1-\eta) \sqrt{\beta^{-2} T(A)} - (\pi/2) |(-i\nabla +A)| \big) \phi_r \Big]_{-} \nonumber \\
&\geq
\tr\Big[ (1-\eta) \sqrt{\beta^{-2} T(A)} - (\pi/2) |(-i\nabla +A)|  \Big]_{-} \nonumber \\
&\geq 
- \tr\Big(- \Big[ (1-\eta)^2 \beta^{-2} T(A) - (\pi/2)^2 |(-i\nabla +A)|^2 \Big]_- \Big)^{1/2}
 \nonumber \\
&  \ge - \beta^{-1} \tr \Big( - \big[ \big(  (1-\eta)^2 - (\beta \pi/2)^2\big) (-i\nabla +A)^2 - 
(1-\eta)^2 |\nabla \times A|\big]_- \Big)^{1/2}
 \nonumber \\
&\geq
- \beta^{-1} \tr \Big( - \big[ 8 \eta (-i\nabla +A)^2 - |\nabla \times A|\big]_- \Big)^{1/2}
 \nonumber \\
&\geq -C 
\beta^{-1} \eta^{-3/2} \int  |\nabla \times A|^2, 
\label{katoBKS}
\end{align}
where we also used the trivial lower bound
for the Pauli operator,
$T(A)\ge (-i\nabla+A)^2 - |\nabla\times A|$, in the fourth line and the special choice of $\eta$ in
the fifth line.
The last inequality in \eqref{katoBKS} is the non-relativistic Lieb-Thirring inequality
for half moments of the negative eigenvalues.
 Note that the bound on this term is consistent with \eqref{eq:stability}.

We proceed to estimate the second term in \eqref{eq:split} using the Lieb-Thirring inequality of Theorem~\ref{thm:LT}.
\begin{align}
\tr&\Big[ \phi_r \big( \eta \sqrt{\beta^{-2} T(A) + \beta^{-4}} - \beta^{-2}{\bf 1}_{\{|x|\leq r\}}
 - V \big)\phi_r \Big]_{-} \nonumber \\
&\geq - C \eta \Bigg\{ 
\eta^{-5/2} \int ( \beta^{-2}1_{\{|x|\leq r\}} + V)^{5/2} + \beta^3 \eta^{-4} \int ( \beta^{-2}{\bf 1}_{\{|x|\leq r\}} + V)^4
\nonumber \\
&\qquad\qquad +
\Big( \int |\nabla \times A|^2 \Big)^{3/4} \eta^{-1} \Big( \int ( \beta^{-2}{\bf 1}_{\{|x|\leq r\}} + V)^4 \Big)^{1/4}
\Bigg\}.
\end{align}
Elementary calculations show that this is also consistent with \eqref{eq:stability}.
\end{proof}

Finally, the expected behaviour for small values of $\beta$ can be obtained by
 the following modified version of Lemma~\ref{thm:CritStability2}.

\begin{lemma}\label{lem:smallbeta} 
Suppose $\beta \in (0, 1/20)$, and $r_1>0$.  Then there exists
a constant $C$ depending only on $r_1$ such that for any $r\le r_1$ we have
\begin{align}
\tr\Big[\sqrt{\beta^{-2} T(A)  + \beta^{-4}}& - \beta^{-2} - \frac{1}{|x|} \cdot {\bf 1}_{\{|x|\leq r\}} - V \Big]_{-}
 \nonumber \\
& \geq - C \Big\{ 1 + \int [V]_+^{5/2} + \int [V]_+^4 + \int |\nabla \times A|^2 \Big\}.
\end{align}
\end{lemma}

\begin{proof}[Proof of Lemma~\ref{lem:smallbeta}] Assuming again $V\ge 0$, 
we first split the energy as follows
\begin{align}
\tr\Big[\sqrt{\beta^{-2} T(A) +  \beta^{-4}} & - \beta^{-2} - \frac{1}{|x|} {\bf 1}_{\{|x|\leq r\}} - V \Big]_{-}  \nonumber \\
&\geq
\frac{1}{2} \tr\Big[\sqrt{\beta^{-2} T(A) + \beta^{-4}} - \beta^{-2} - 2V \Big]_{-} \nonumber \\
&\quad+
\frac{1}{2}\tr\Big[\sqrt{\beta^{-2} T(A) + \beta^{-4}} - \beta^{-2} - \frac{2}{|x|} {\bf 1}_{\{|x|\leq r\}} \Big]_{-} .
\end{align}
The desired estimate for the term with $V$ follows from Theorem~\ref{thm:LT}, 
so it suffices to consider the second term with the Coulomb potential.
Following the proof of Theorem~\ref{thm:LT} 
we split this as
\begin{align}\label{eq:splitSmallLarge}
\tr\Big[\sqrt{\beta^{-2} T(A) + \beta^{-4}} & - \beta^{-2} - \frac{2}{|x|} {\bf 1}_{\{|x|\leq r\}} \Big]_{-} \nonumber \\
&\geq
\frac{1}{2}\tr\Big[\sqrt{\beta^{-2} T(A) + \beta^{-4}} - \beta^{-2} - 
\frac{4}{|x|} {\bf 1}_{\{4\beta^2 \leq |x|\leq r\}} \Big]_{-} \nonumber \\
&\quad + \frac{1}{2} \tr\Big[\sqrt{\beta^{-2} T(A) + \beta^{-4}} - \beta^{-2} 
- \frac{4}{|x|} {\bf 1}_{\{|x|\leq 4\beta^2\}} \Big]_{-}.
\end{align}
The first term on the right, where the potential is bounded by $\beta^{-2}$, we can
 estimate similarly to \eqref{eq:SqrtToPauli} as
\begin{align}
\tr\Big[\sqrt{\beta^{-2} T(A) + \beta^{-4}} - \beta^{-2} - \frac{4}{|x|} 1_{\{4\beta^2 \leq |x|\leq r\}} \Big]_{-} 
\geq
\tr\Big[ c_0 T(A) - \frac{8}{|x|} {\bf 1}_{\{|x|\leq r\}} \Big]_{-}.
\end{align}
This is a Pauli operator with a Coulomb singularity and 
is known to be bounded from below by a constant depending only on
the upper bound on $r$, see the (proof of)  Lemma 2.1 in \cite{ES2}
with the choice of $Z= 8c_1^{-1}$
(see also \cite[Equation (4.8)]{EFS3} with a special choice of the constants).

For the second term in \eqref{eq:splitSmallLarge} we use the BKS inequality
\eqref{eq:BKS} and estimate
\begin{align}
\tr\Big[& \sqrt{\beta^{-2} T(A) + \beta^{-4}}  - \beta^{-2} - \frac{4}{|x|} {\bf 1}_{\{|x|\leq 4\beta^2\}} \Big]_{-}  \\
&\geq - \tr\Big(- \Big[\beta^{-2} T(A) + \beta^{-4} - \big(\beta^{-2} +
\frac{4}{|x|} {\bf 1}_{\{|x|\leq 4\beta^2\}}\big)^2 \Big]_-\Big)^{1/2} \nonumber \\
&= - \tr\Big(- \Big[\beta^{-2} T(A) - \beta^{-2} \frac{8}{|x|} 1_{\{|x|\leq 4\beta^2\}} 
- \frac{16}{|x|^2} {\bf 1}_{\{|x|\leq 4\beta^2\}}\Big]_- \Big)^{1/2} \nonumber \\
&\geq - \tr\Big(-\Big[\beta^{-2} T(A) - \frac{48}{|x|^2} {\bf 1}_{\{|x|\leq 4\beta^2\}} 
\Big]_-\Big)^{1/2} \nonumber \\
&\geq - \tr\Big( - \Big[ 400\big\{ (-i\nabla+A)^2 - |\nabla \times A|\big\} - \frac{48}{|x|^2} \Big]_-\Big)^{1/2}, \nonumber 
\end{align}
where we used $\beta<1/20$, $T(A)\ge0$ and that $T(A)\ge  (-i\nabla+A)^2 - |\nabla \times A|$.
We now use the Hardy inequality, $(4|x|)^{-2}\le (-i\nabla+A)^2$  to continue this estimate as
\begin{align}
\geq - \tr \Big( -\Big[ 208 (-i\nabla+A)^2 - 400 |\nabla \times A| \Big]_-
\Big)^{1/2} \geq - C \int |\nabla \times A|^2
\end{align}
by the non-relativistic Lieb-Thirring inequality for the half moments.
 This finishes the proof of Lemma~\ref{lem:smallbeta}.
\end{proof}

As explained previously, the results of Lemma~\ref{lem:smallbeta} and Lemma~\ref{thm:CritStability2}
 combine to imply Theorem~\ref{thm:CritStability}. Therefore the proof of
 Theorem~\ref{thm:CritStability} is finished.
\end{proof}

\end{document}